\documentclass[a4paper,11pt]{article}
\usepackage[top=2cm,bottom=2cm,left=2cm,right=2cm]{geometry}

\usepackage{amsmath,amsthm,amsfonts,amssymb,bm}
\usepackage{graphicx}
\graphicspath{{Figures/}}
\usepackage{enumerate}
\allowdisplaybreaks
\usepackage[colorlinks=true, allcolors=blue]{hyperref}
\usepackage[sectionbib]{natbib}

\usepackage{authblk}

\newtheorem{thm}{Theorem}[section]
\newtheorem{lemma}{Lemma}[section]

\newtheorem{remark}{Remark}
\newtheorem{ex}{Example}

\newcommand{\R}{\mathbb{R}}
\newcommand{\E}{\mathbb{E}}

\renewcommand{\vec}[1]{\mathbf{#1}}
\newcommand{\parvec}[1]{\bm{#1}}

\newcommand{\xivec}{\mbox{\boldmath $\xi$}}
\newcommand{\muvec}{\mbox{\boldmath $\mu$}}
\newcommand{\sigmat}{\bm{\Sigma}}

\newcommand{\se}[1]{$\scriptstyle(#1)$}
\newcommand{\best}[1]{\textbf{#1}}

\title{Classification Using Global and Local Mahalanobis Distances}
\author[1]{Annesha Ghosh}
\author[1]{Anil K. Ghosh}
\author[2]{Rita SahaRay}
\author[3]{Soham Sarkar\footnote{Corresponding author. Email: {sohamsarkar@isid.ac.in}}}
\affil[1]{Theoretical Statistics and Mathematics Unit, Indian Statistical Institute, \protect\\ 203, B.\ T.\ Road, Kolkata 700108, India.}
\affil[2]{Department of Statistics and Applied Probability, University of California, Santa Barbara, \protect\\ CA 93106, USA.}
\affil[3]{Statistical Sciences Unit, Indian Statistical Institute, Delhi Centre, \protect\\ 7, S.\ J.\ S.\ Sansanwal Marg, New Delhi 110016, India.}

\date{}

\begin{document}
\maketitle
\begin{abstract}
We propose a novel semiparametric classifier based on Mahalanobis distances of an observation from the competing classes. Our tool is a generalized additive model with the logistic link function that uses these distances as features to estimate the posterior probabilities of different classes. While popular parametric classifiers like linear and quadratic discriminant analyses are mainly motivated by the normality of the underlying distributions, the proposed classifier is more flexible and free from such parametric modeling assumptions. Since the densities of elliptic distributions are functions of Mahalanobis distances, this classifier works well when the competing classes are (nearly) elliptic. In such cases, it often outperforms popular nonparametric classifiers, especially when the sample size is small compared to the dimension of the data. To cope with non-elliptic and possibly multimodal distributions, we propose a local version of the Mahalanobis distance. Subsequently, we propose another classifier based on a generalized additive model that uses the local Mahalanobis distances as features. This nonparametric classifier usually performs like the Mahalanobis distance based semiparametric classifier when the underlying distributions are elliptic, but outperforms it for several non-elliptic and multimodal distributions. We also investigate the behaviour of these two classifiers in high dimension, low sample size situations. A thorough numerical study involving several simulated and real datasets demonstrate the usefulness of the proposed classifiers in comparison to many state-of-the-art methods.

\end{abstract}

\section{Introduction}\label{sec:intro}
Mahalanobis distance \citep[MD,][]{pcm} plays a major role in various statistical analyses of multivariate data. Some examples of its widespread applications include computing the distance between two populations, measuring the centrality of a multivariate observation with respect to a data cloud or a probability distribution, two-sample test based on Hotelling's $T^2$ type statistics, classification using linear and quadratic discriminant analyses, and clustering based on Gaussian mixture models. The density of an elliptic distribution \citep[e.g.,][]{fang2018symmetric} is a function of the Mahalanobis distance. So, in a classification problem involving two or more elliptic distributions, the Bayes classifier \citep[e.g.,][]{hastie2009elements} turns out to be a function of the Mahalanobis distances of the observation from the competing classes. This result provides the motivation for constructing a classifier based on these distances. To demonstrate the utility of such a classifier, let us consider a simple example of binary classification.
\begin{ex}\label{example1}
Each of the two competing classes is an equal mixture of two uniform distributions. Class-1 is a mixture of $U_d(0,1)$ and $U_d(2,3)$, whereas Class-2 is a mixture of $U_d(1,2)$ and $U_d(3,4)$. Here, $U_d(a,b)$ denotes the {$d$-dimensional} uniform distribution over the region $\{\vec x \in \R^d:\,\, a \le \|\vec x\|\le b\}$.
\end{ex}

\begin{figure}
\centering
\includegraphics[height=2.0in]{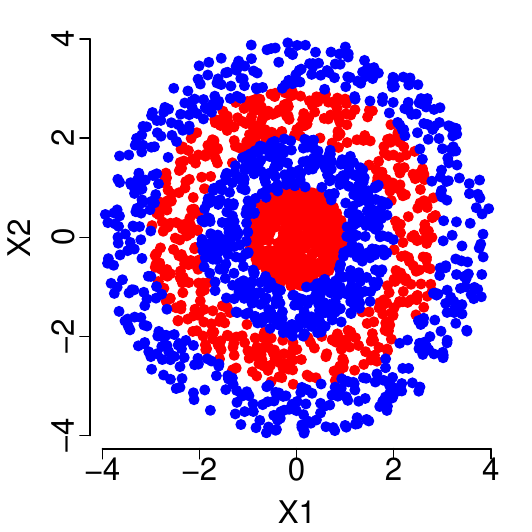}
\caption{Scatter plot of 2000 observations (1000 from each class) from the two classes in Example~\ref{example1} with $d=2$. Red (respectively, blue) represents observations from Class-1 (respectively, Class-2).\label{fig:scatter_ex1}}
\vspace{-0.2in}
\end{figure}

A scatter plot of this dataset for $d=2$ is shown in Figure~\ref{fig:scatter_ex1}. We consider this example with $d=2,4$ and $6$, and in each case, we form training and test sets of sizes $200$ and $10000$, respectively, taking an equal number of observations from the two classes. We repeat the data generation process $100$ times and compute the average empirical misclassification rates of several popular parametric and nonparametric classifiers on the test sets. These misclassification rates are reported in Table~\ref{tab:example1}. In this example, the competing classes have disjoint supports, so the Bayes risk is zero. But, the Bayes classifier is highly non-linear. So, as expected, linear discriminant analysis \citep[LDA, e.g.,][]{anderson1962introduction} and other linear classifiers like logistic regression, GLMNET \citep[e.g.,][]{friedman2010regularization} and linear support vector machines \citep[SVM Linear, e.g.,][]{duda2006pattern} have misclassification rates close to 50\%. Performances of quadratic discriminant analysis \citep[QDA, e.g.,][]{anderson1962introduction} and nonparametric methods like the $k$-nearest neighbour classifier \citep[$k$NN, e.g.,][]{cover1967nearest}, kernel discriminant analysis \citep[KDA, e.g.,][]{hand1982kernel}, classification tree \citep[CART, e.g.,][]{breiman2017classification} and random forest \citep{breiman2001random} are only marginally better. All these state-of-the-art classifiers fail to yield satisfactory performance. Nonlinear SVM with radial basis function kernel \citep[SVM RBF, e.g.,][]{christianini2003support, smola1998learning} has somewhat better performance, but even this classifier has misclassification rates close to 35\% for $d=4$ and $d=6$. Apart from the curse of dimensionality, one major drawback of the nonparametric methods is that even if we have some structural information about the underlying distributions, that useful information cannot be utilized during the construction of the classification rule. If we can extract that information and use it judiciously, the resulting classifier can have improved performance. It is known that for elliptic distributions, MD carries substantial information about the density. Keeping that in mind, we construct a classifier based on MD that leverages the underlying structural information (details given in the next section). This classifier (henceforth, referred to as the MD classifier) does not involve any density estimation, but still can be viewed as a generalization of LDA and QDA. The proposed classifier has excellent performance in Example~\ref{example1}, where it outperforms all other classifiers considered here. In particular, for $d=4$ and $6$, while all other classifiers have misclassification rates higher than 30\%, the misclassification rates of the MD classifier are less than 10\%.

\begin{table}[h!]
\setlength{\tabcolsep}{0.05in}
    \centering
    \small
    \begin{tabular}{c|ccccccccccccc}
         $d$ & LDA & QDA & Logistic & GLMNET & $k$NN & KDA & CART & Random  & SVM & SVM &MD\\ 
         & & &Regression &   &    & & &Forest & Linear & RBF & (proposed) \\ \hline
         $2$ & 49.58 & 41.92 & 49.59 & 49.57  & 12.75 & 19.71 & 28.86 & 14.22  & 45.09 & 7.84 & \best{6.85}\\ 
          &  \se{0.11} & \se{0.25} & \se{0.11} & \se{0.11} & \se{0.20} & \se{0.23} & \se{0.43} & \se{0.21} & \se{0.12} & \se{0.13} & \se{0.27}\\ \hline
         $4$ & 49.43 & 42.11 & 49.45 & 49.37  & 32.80 & 36.72& 38.85 & 33.08 & 46.39& 33.44 & \best{7.75} \\
         & \se{0.10} & \se{0.14} & \se{0.10} & \se{0.10}  & \se{0.31}& \se{0.24} & \se{0.26} & \se{0.12} & \se{0.18} & \se{0.23} & \se{0.18} \\\hline
         $6$  &49.12 & 42.29 & 49.16 & 49.17 & 39.56 & 41.41 & 39.84 & 34.64 & 46.55 &36.98 & \best{8.46}\\ 
         & \se{0.11} & \se{0.12} & \se{0.11} & \se{0.11}  & \se{0.19} & \se{0.20} & \se{0.25} & \se{0.12} & \se{0.17} & \se{0.16} & \se{0.18} \\\hline
    \end{tabular}
    
    Boldface character signifies the best result in each case.
    \caption{Average misclassification rates (in \%) of different classifiers in Example~\ref{example1} for varying choices of $d$. The standard errors are reported in the next line within brackets in a smaller font.\label{tab:example1}}
\end{table}

The description of the MD classifier is given in the next section, where we also analyze a few other simulated datasets involving elliptic distributions to evaluate its performance. However, when the underlying distributions are non-elliptic, in particular, if they are multimodal in nature, the MD classifier may not work well. To take care of this problem, we propose a local version of MD and develop a classifier based on those local distances. The description of this classifier is given in Section~\ref{sec:LMD classifier}. This local Mahalanobis distance (LMD) involves a tuning parameter that controls the degree of localization. For higher values of this parameter, local Mahalanobis distances behave like usual Mahalanobis distances, and the classifier behaves like the MD classifier. But, for smaller values of the tuning parameter, it performs like the nonparametric methods. So, for a suitable choice of this parameter, the proposed method can perform well both for elliptic and non-elliptic distributions. We use the bootstrap method to select this parameter and investigate the performance of the resulting classifier on several simulated datasets. In Section~\ref{sec:HDLSS}, we study the high-dimensional behaviour of our proposed classifiers and observe that unlike many popular nonparametric classifiers, our proposed methods can perform well even in high dimension, low sample size (HDLSS) situations. Several benchmark datasets are analyzed in Section~\ref{sec:real data} to compare the performance of the proposed classifiers with some state-of-the-art classifiers. Section~\ref{sec:conclusion} contains a brief summary of the work and ends with a discussion on related issues. All proofs and mathematical details are given in the Appendix. 

\section{Classifier Based on Mahalanobis Distances\label{sec:MD classifier}}

Consider a $J$-class classification problem, where the $j$-th class ($j=1,\ldots,J$) has probability density function $f_j$ with location $\parvec\mu_j$ and scatter matrix $\parvec\Sigma_j$. For any observation $\vec x$, its Mahalanobis distance from the $j$-th class is given by $\delta_j({\bf x})=\{({\bf x}-\muvec_j)^{\top}\sigmat_j^{-1}({\bf x}-\muvec_j)\}^{1/2}$, where $\vec a^{\top}$ denotes the transpose of the vector $\vec a$. Note that if $f_j$ is the density\ of an elliptically symmetric distribution \citep[e.g.,][]{fang2018symmetric}, then its functional form is given by
\[
f_j({\bf x})= C_j |\sigmat_j|^{-1/2} \phi_j (\delta_j({\bf x})),
\]
where $|\sigmat_j|$ is the determinant of $\sigmat_j$, $C_j$ is a positive constant, and $\phi_j:{\mathbb R}_+ \rightarrow {\mathbb R}_+$ is a non-negative function. Therefore, if the competing classes are elliptically symmetric and $\pi_j$ is the prior probability of the $j$-th class, then the logarithm of the ratio of the posterior probabilities of the $j$-th and the $j^\prime$-th classes ($1 \le j \neq j^\prime \le J$) can be expressed as 
\[
\log\left(\frac{p(j \mid \vec x)}{p(j^\prime \mid \vec x)}\right) = (\alpha_{j}-\alpha_{j^\prime}) + \psi_j(\delta_j({\bf x})) - \psi_{j^\prime}(\delta_{j^\prime}({\bf x})),
\]
where $\psi_j(t)=\log(\phi_j(t))$ and $\alpha_j = \log(C_j \pi_j |\sigmat|^{-1/2}_j)$. So, the posterior probabilities of different classes follow a generalized additive model \citep[GAM, e.g.,][]{hastie1990generalized, wood2017generalized} with the logistic link function involving Mahalanobis distances $\Delta({\bf x})=(\delta_1({\bf x}),\ldots,\delta_J({\bf x}))$ as the covariates.  We state this as a theorem below.

\begin{thm}\label{thm:MD_GAM}
If the densities of $J$ competing classes are elliptically symmetric, their posterior probabilities satisfy an additive logistic regression model given by
\[
p(j \mid \vec x) = \frac{\exp\Big(g_j\big(\Delta(\vec x)\big)\Big)}{1+\sum_{k=1}^{J-1}\exp\Big(g_k\big(\Delta(\vec x)\big)\Big)} \text{ for } j=1,\ldots,J-1 \text{ and } p(J \mid {\bf x})= \frac{1}{1+\sum_{k=1}^{J-1}\exp\Big(g_k\big(\Delta(\vec x)\big)\Big)},
\]
where $g_j\big(\Delta({\bf x})\big)=\sum_{k=1}^{J} g_{jk}\big(\delta_k({\bf x})\big)$ is an additive function of $\Delta(\vec x)=\big(\delta_1(\vec x),\ldots,\delta_J(\vec x)\big)$.
\end{thm}

Motivated by the above result, for constructing a classifier, we start by estimating the empirical Mahalanobis distances from the training data $\{{\bf x}_{ji}:~j=1,\ldots,J,~i=1,\ldots,n_j\}$, where ${\bf x}_{ji}$ denotes the $i$-th observation from the $j$-th class. We use this dataset to estimate the location vectors and scatter matrices of different classes. Usually, one uses moment-based estimates, but robust estimates like minimum covariance determinant \citep[MCD, e.g.,][]{rousseeuw1999fast} or minimum volume ellipsoid \citep[MVE, e.g.,][]{van2009minimum} estimates can also be used. After getting these estimates ${\widehat \muvec}_j$ and ${\widehat\sigmat}_j$ for $j=1,\ldots,J$, for each  ${\bf x}_{ji}$, we estimate the vector of Mahalanobis distances by ${\widehat \Delta}({\bf x}_{ji})=\big(\widehat\delta_1({\bf x}_{ji}),\ldots,{\widehat\delta_J}({\bf x}_{ji})\big)$, where ${\widehat\delta}_j({\bf x}) = \big\{({\bf x}-{\widehat\muvec}_j)^{\top} {\widehat\sigmat}_j^{-1} ({\bf x}-{\widehat\muvec}_j)\big\}^{1/2}$ for $j=1,\ldots,J$. Pointwise convergence of $\widehat\delta_j$ to $\delta_j$ follows from the convergence of $\widehat\muvec_j$ and $\widehat\sigmat_j$ to their population counterparts. In fact, uniform convergence of $\widehat\delta_j$ over any compact set holds for consistent estimators $\widehat\muvec_j$ and $\widehat\sigmat_j$ of $\muvec$ and $\sigmat$, respectively. This follows from the result given below.
\begin{lemma}\label{lemma:MD_uniform_convergence}
Let $\vec X_1,\ldots,\vec X_n$ be independent copies of a $d$-dimensional random vector $\vec X$, which has finite second moments. Define $\muvec = E(\vec X)$ and $\sigmat = Var(\vec X)$, and suppose that $\sigmat$ is positive definite. Let $\widehat\muvec$ and $\widehat\sigmat^{-1}$ be $\sqrt{n}$-consistent estimators of $\muvec$ and $\sigmat^{-1}$, respectively, i.e., $\|\widehat\muvec-\muvec\| = O_P(n^{-1/2})$ and $\|\widehat\sigmat^{-1}\mkern-2mu-\sigmat^{-1}\|_{\rm F} = O_P(n^{-1/2})$, where $\|\cdot\|_{\rm F}$ is the Frobenius norm. Then, for any compact set $\mathcal C \subset \R^d$,
\[
\sup_{\vec x \in \mathcal C} \left\|(\vec x- \widehat\muvec)^\top \widehat\sigmat^{-1}(\vec x-\widehat\muvec) - (\vec x- \muvec)^\top \sigmat^{-1}(\vec x-\muvec)\right\| = O_P(n^{-1/2}).
\]
\end{lemma}
\begin{remark}\label{remark:MD_uniform_convergence}
The assumption of $\vec X$ having finite second moment in Lemma~\ref{lemma:MD_uniform_convergence} is not necessary. The same proof holds with $\muvec$ and $\sigmat$ being the location vector and scatter matrix (not necessarily moment based), as long as we can find $\sqrt{n}$-consistent estimators for $\muvec$ and $\sigmat^{-1}$. Moreover, if we have consistent estimators but of different rates of convergence, then also the result follows. In this case, the uniform rate of convergence will be the slower of the two convergence rates (of $\widehat\muvec$ and $\widehat\sigmat^{-1}$).
\end{remark}

In the case of a high-dimensional problem, the estimated MDs $\widehat\Delta(\vec x_{ji})$ can provide a low-dimensional view of the class separability. This visualization is particularly useful for binary classification. To demonstrate this, we show in Figure~\ref{fig:MD_plot_ex1} the scatter plot of the ${\widehat \Delta}({\bf x}_{ji})$'s for the classification problem in Example~\ref{example1} with $d=4$ and $d=6$. The separability of the two classes is clearly visible from this figure.

\begin{figure}[h]
\begin{center} 
\includegraphics[width=0.35\linewidth,,height=2in]{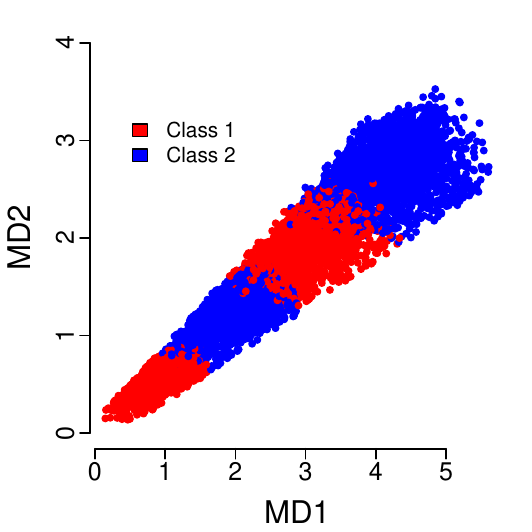}
\hspace{0.1\linewidth}
\includegraphics[width=0.35\linewidth,height=2in]{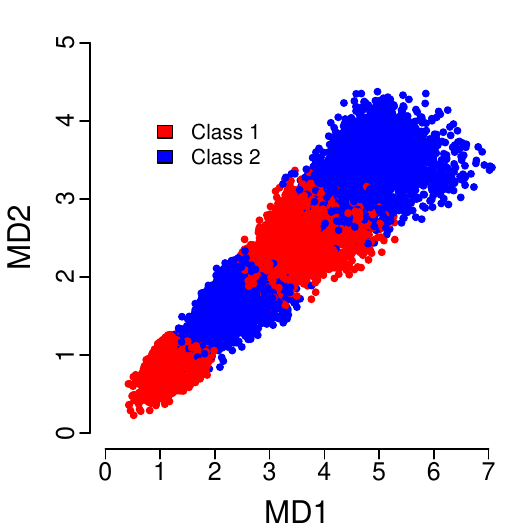}
\caption{Scatter plot of estimated Mahalanobis distances in Example~\ref{example1} for $d=4$ (left) and $d=6$ (right). Here, MD1 and MD2 denote estimated Mahalanobis distances with respect to Class-1 and Class-2 (i.e., ${\widehat \delta}_1(\cdot)$ and ${\widehat \delta}_2(\cdot)$), respectively.\label{fig:MD_plot_ex1}}
\end{center}
\vspace{-0.1in}
\end{figure}

Motivated by Theorem~\ref{thm:MD_GAM}, we use the $\widehat\Delta(\vec x_{ji})$'s as extracted features and fit a GAM to get estimates of the additive functions $\widehat g_1,\ldots,\widehat g_J$, and subsequently, those of the posterior probabilities  $\widehat p(1 \mid \vec x),\ldots,\widehat p(J \mid \vec x)$. Finally, an observation $\vec x$ is assigned to the class having the highest estimated posterior. We have already seen that this classifier has an excellent performance in Example~\ref{example1}. In the next subsection, we evaluate its empirical performance on a few more simulated datasets. 

\subsection{Empirical Performance of the MD Classifier\label{sec:MD_simulation}}
 
Here, we consider some classification problems involving two or more elliptic distributions. For all these examples, we form the training set by generating 100 observations from each class, while 1000 (5000 for two-class problems) observations from each class are generated to form the test set. Each experiment is repeated 100 times to compute the average test set misclassification errors and the corresponding standard errors for different classifiers. These results are reported in Table~\ref{tab:simulation_MD_2-8} for $d=2,4$ and $6$. Like Example~\ref{example1}, here also we compare the performance of the MD classifier with some popular classifiers. Some of these competing classifiers involve tuning parameters, which are chosen using 5-fold cross-validation. All classification methods are implemented using \texttt{R} codes. The codes for our proposed methods are available at \url{https://github.com/Agy-code/MD-LMD-Classifier}.

\begin{ex}\label{example2}
We consider two multivariate normal distributions $N_d({-0.3}{\bf 1}_d,{\bf I}_d)$ and $N_d({0.3}{\bf 1}_d,{\bf I}_d)$ differing only in their locations. Here, ${\bf 1}_d$ denotes the $d$-dimensional vector with all elements equal to $1$, and ${\bf I}_d$ denotes the $d \times d$ identity matrix.
\end{ex}
Here, the Bayes classifier is linear, and this is an ideal setup for LDA to perform well. So, as expected, LDA has the lowest misclassification rate in this example, and all linear classifiers perform better than nonlinear classifiers. Nevertheless, all classifiers have satisfactory performances.

\begin{ex}\label{example3}
We consider two multivariate normal distributions $N_d({\bf 0}_d,{\bf I}_d)$ and $N_d({\bf 0}_d,5 {\bf I}_d)$, which differ only in their scales. Here, ${\bf 0}_d$ is the $d$-dimensional vector with all elements equal to $0$.
\end{ex}
In this example, the Bayes classifier is a quadratic classifier, and this setup is ideal for QDA. So, QDA has the best performance here. The MD classifier has the second-best performance, closely followed by {the} nonlinear SVM. Among the others, Random Forest has a somewhat satisfactory performance. The rest of the classifiers perform very poorly.

\begin{ex}\label{example4}
Consider two elliptical distributions, one of which is $U_{d,\sigmat}(1,2)$ and the other one is an equal mixture of $U_{d,\sigmat}(0,1)$ and $U_{d,\sigmat}(2,3)$. Here, $U_{d,\sigmat}(a,b)$ denotes the $d$-dimensional uniform distribution over the region $\{{\bf x} \in {\mathbb R}^d: a \le \|\sigmat^{1/2}{\bf x}\|\le b\}$. The matrix $\sigmat$ has all diagonal elements equal to $1$ and all off-diagonal elements equal to $0.5$.
\end{ex}
Like Example~\ref{example1}, here also the MD classifier performs much better than its competitors. While it has misclassification rates close to 5\%, all other competing methods have much higher misclassification rates, especially for $d=4$ and $d=6$.

\begin{table}[!t]
\setlength{\tabcolsep}{0.05in}
  \centering
    {\small
    \begin{tabular}{c|c|cccccccccccc}
       Data &$d$&Bayes & LDA & QDA & Logistic & GLM & $k$NN & KDA & CART & Random & SVM & SVM & MD\\ 
        set  & &       &    & & Reg.    & NET &    & & & Forest & Linear & RBF     & (proposed) \\ \hline
        & 2 &33.53 & \best{33.89} & 34.18 & \best{33.89} & 34.39 &  35.71 & 35.12 & 38.08 &39.42 & 33.93 & 36.11 & 35.14\\
        & &  \se{0.05} &\se{0.07} &\se{0.08}& \se{0.07} &\se{0.14} & \se{0.22} &\se{0.29} &\se{0.24} &\se{0.17} &\se{0.06} &\se{0.29} &\se{0.15}\\ \cline{2-14}
         Ex 2 & 4& 27.48 & \best{28.09} & 28.97 & 28.10 & 28.50 & 29.27 & 28.93 & 34.92 & 31.29 & 28.13 & 29.36 & 29.59\\  
         & &\se{0.04}&\se{0.06} &\se{0.08} & \se{0.06}&  \se{0.10}& \se{0.12}& \se{0.16}& \se{0.21}&\se{0.11} & \se{0.06}& \se{0.20}&\se{0.10}\\\cline{2-14}
           &6 & 23.14 & \best{24.02} & 25.47 & 24.08 & 24.64 & 25.32 & 24.74 & 34.16 & 26.84 & 24.14 & 24.77 & 26.00\\  
         & &\se{0.05} &\se{0.07} &\se{0.10} & \se{0.07}& \se{0.12}& \se{0.11} & \se{0.11}& \se{0.20}& \se{0.10}&\se{0.07} &\se{0.11}&\se{0.11} \\ 
      \hline
 & 2 & 23.23 & 48.95 & \best{23.61} & 48.97 & 49.02 & 26.64 & 25.87 & 26.76 & 27.78 & 41.55 & 24.95 & 23.93 \\
        & &\se{0.04} & \se{0.12} & \se{0.05} & \se{0.12} & \se{0.11} & \se{0.16} & \se{0.14} & \se{0.25} & \se{0.12} & \se{0.11} & \se{0.14} & \se{0.07} \\ \cline{2-14}
              Ex 3 &4 &14.21 & 48.32 & \best{15.08} & 48.36 & 48.38 & 23.07 & 21.95 & 24.07 & 17.76 & 41.85 & 15.60 & 15.27 \\ 
     & & \se{0.03} & \se{0.14} & \se{0.05} & \se{0.14} & \se{0.14} & \se{0.15} & \se{0.13} & \se{0.26} & \se{0.10} & \se{0.10} & \se{0.11} & \se{0.06} \\ \cline{2-14}
          &6 & 9.12 & 47.84 & \best{10.42} & 47.90 & 47.91 & 23.94 & 22.56 & 23.86 & 13.56 & 42.05 & 10.73 & 10.70 \\ 
          & &\se{0.03} & \se{0.12} & \se{0.04} & \se{0.12} & \se{0.12} & \se{0.14} & \se{0.10} & \se{0.26} & \se{0.12} & \se{0.11} & \se{0.10} & \se{0.08} \\ \hline
        & 2 & 0.00 & 50.32 & 48.67 & 50.31 & 50.27 & 9.49 & 8.42 & 20.93 & 10.56 & 53.04 & 5.82 & \best{4.72}\\
        & & \se{0.00} &\se{0.08} & \se{0.50} & \se{0.08} & \se{0.08} 
          & \se{0.16} &\se{ 0.12} & \se{0.29} 
           & \se{0.14} & \se{0.08} & \se{ 0.11} & \se{0.24}\\ \cline{2-14}
        Ex 4 &4 &0.00 & 50.11 & 51.29 & 50.10 & 50.11 & 25.30 & 31.40 & 38.69 & 25.84 & 51.26 & 16.60 & \best{5.18} \\ 
        & &\se{0.00} & \se{0.07} & \se{0.33} & \se{0.07} & \se{0.08} &   \se{0.21} & \se{0.29} & \se{0.39} & \se{0.18} & \se{0.11} & \se{0.19} & \se{0.14}\\ \cline{2-14}
          &6 & 0.00 & 50.02 & 52.77 & 50.01 & 50.00& 32.95 & 40.93 & 42.53 & 31.04 & 50.80 & 25.61 & \best{4.80}\\ 
           & & \se{0.00}& \se{0.07} & \se{0.26} & \se{0.07} & \se{0.07}& \se{0.21} & \se{0.37} &\se{0.25}&\se{0.21}& \se{0.11} & \se{0.22} &\se{0.15}\\ \hline
        & 2 & 36.17& 49.92 & 45.72 & 49.89 & 49.89  & 41.79 & 42.05 & 42.80 & 43.36 & 49.73 & 40.02 & \best{37.82}\\
        & & \se{0.04} & \se{0.11} & \se{0.59} & \se{0.11} & \se{0.10} &  \se{0.21} & \se{0.25} & \se{0.24}& \se{0.14} & \se{0.26} & \se{0.26} & \se{0.13}\\ \cline{2-14}
        Ex 5 &4 & 30.48 & 50.02 & 43.41 & 49.98 & 49.95 & 39.86 & 40.08 & 41.12 & 37.85& 48.96 & 34.83 & \best{32.87}\\ 
        & & \se{0.05} & \se{0.10} &\se{0.67} &\se{0.10} & \se{0.09} &\se{0.16} &\se{ 0.30} &\se{0.24} &\se{0.15} &\se{ 0.21} &\se{ 0.22} &\se{ 0.14}\\ \cline{2-14}
          &6 &  27.19 & 49.94 & 41.40 & 49.90 & 49.89 & 39.79 & 40.49 & 41.10 & 35.30 & 48.68 & 32.58 & \best{30.95}\\ 
           & &\se{0.04} &\se{0.10} &\se{0.67}& \se{0.10}& \se{0.10} & \se{0.12}& \se{0.41}& \se{0.26} & \se{0.13} &\se{0.17} &\se{0.23}& \se{0.17}\\ \hline
        & 2 & 47.32 & 66.70 & 66.47 & 66.68 & 66.71  & 56.44 & 55.83 & 58.88 & 57.62 & 69.01 & 54.87 & \best{51.65}\\
        & & \se{0.08} &\se{0.24}& \se{0.30} &\se{0.24} &\se{0.22}& \se{0.19} & \se{0.23} & \se{0.22} & \se{0.12} &\se{0.19} & \se{0.20} &\se{0.13} \\ \cline{2-14}
          Ex 6 & 4 &  47.51 & 66.39 & 66.09 & 66.41 &66.46  & 59.23 & 59.43 & 61.30 & 59.19 & 67.69 & 58.00 & \best{53.07}\\ 
       & &  \se{0.07} & \se{0.24} &\se{0.31} &\se{0.24} & \se{0.22} &  \se{0.12} & \se{0.14} & \se{0.16} & \se{0.12} & \se{0.23} & \se{0.17} & \se{0.13}\\ \cline{2-14}
          &6 & 47.29 & 66.41 & 65.95 & 66.35 &66.35& 60.56 &61.02& 62.13 & 59.40 &67.17 &59.40 & \best{53.83}  \\ 
          & & \se{0.08} & \se{0.22} & \se{0.29} & \se{0.22} & \se{0.22} & \se{0.12} & \se{0.20} & \se{0.16} & \se{0.10} & \se{0.24} & \se{0.14} & \se{0.13}\\ \hline
        & 2 & 44.77 & 65.16 & \best{45.71} & 65.17 & 65.27 &  50.61 & 52.40 & 56.18 & 53.41 & 61.87 & 48.53 & 46.52\\
        & & \se{0.08} & \se{0.32} & \se{0.10} & \se{0.32} & \se{0.31} & \se{0.16} & \se{0.29} &\se{0.26} & \se{0.13} & \se{0.42} & \se{0.14} & \se{0.11}\\ \cline{2-14}
          Ex 7 & 4 & 29.68 & 63.85 & \best{31.52} & 63.87 & 63.81 &  41.97 & 44.80 & 50.46 & 39.04 & 57.77 & 36.48 & 32.32\\ 
       & & \se{0.07} &\se{0.28}& \se{0.11} & \se{0.28} & \se{0.28} &  \se{0.14} & \se{0.09} & \se{0.28} & \se{0.15} &\se{0.48} & \se{0.14} & \se{0.12}\\ \cline{2-14}
          &6 & 23.24 & 62.68 & \best{25.96} & 62.73 &62.91 & 40.05 & 43.43 & 49.26 & 34.07 & 56.05 & 31.64 & 26.82 \\ 
          & & \se{0.08}& \se{0.26} & \se{0.11} & \se{0.26} & \se{0.28} &  \se{0.16} &\se{0.10}& \se{0.31} & \se{0.13} & \se{0.41} & \se{0.17} & \se{0.12}\\ \hline   
          & 2 & 33.60 & 49.85 & 41.20 & 47.82 & 47.64& 40.37 & 37.72 & 38.72 & 39.33 & 48.03 & 37.49 &\best{35.01}\\
        & & \se{0.04} & \se{0.06} & \se{0.39} & \se{0.14} & \se{0.17} & \se{0.16} & \se{0.18} & \se{0.24} & \se{0.10} & \se{0.16} & \se{0.21} & \se{0.09}\\ \cline{2-14}
          Ex 8 & 4 & 28.72 & 49.85 & 44.86 & 47.94 & 47.51 & 40.80 & 37.41 & 38.82 & 35.46 & 46.80 & 33.81 & \best{30.50}\\ 
       & & \se{0.05} &\se{0.06} & \se{0.52} & \se{0.12} & \se{0.14} & \se{0.21} & \se{0.18} & \se{0.25} & \se{0.12} & \se{0.20} & \se{0.23} & \se{0.13}\\ \cline{2-14}
          &6 &25.54 & 49.79 & 47.46 & 48.08 & 47.28 & 41.49 & 38.70 & 38.43 & 33.59 & 46.04 & 32.88 &\best{28.30}\\   
          & & \se{0.04}& \se{0.07} & \se{0.29} & \se{0.10} & \se{0.15} &  \se{0.23} & \se{0.20}& \se{0.23} & \se{0.14} & \se{0.19} & \se{0.29} & \se{0.29}\\ \hline
    \end{tabular}

    Boldface character signifies the best result in each case.}
    \caption{Average misclassification rates (in \%) of different classifiers in Examples~\ref{example2} -- \ref{example8} for varying $d$. The corresponding standard errors are reported in the next line within brackets in a smaller font.\label{tab:simulation_MD_2-8}}

\vspace{-0.1in}
\end{table}

\begin{ex}\label{example5}
Here, the underlying distributions of the two competing classes are $N_d({\bf 0}_d,3{\bf I}_d)$ and $t_{3,d}({\bf 0}_d,{\bf I}_d)$, the standard multivariate $t$ distribution with $3$ degrees of freedom.
\end{ex}
These two distributions have the same locations and covariance matrices, while their shapes are different. In this example, the MD classifier outperforms all other classifiers. Only nonlinear SVM and Random Forest have somewhat competitive performances.

\begin{ex}\label{example6}
We consider three spherical distributions \citep[e.g.,][]{fang2018symmetric}. A spherical distribution is completely specified by the distribution of its radius $R$. The radius $R$ follows the uniform distribution $U(0,10)$ for Class-1 and the normal distribution $N(5.5,\sigma^2)$ for Class-2. For Class-3, it is distributed as $cY$, where $Y$ follows a beta distribution with parameters $0.5$ and $0.5$. We choose $\sigma^2$ and $c$ in such a way that $E(R^2)$ is the same in all three cases.
\end{ex}
This is a difficult classification problem, where all competing classes have the same location vector and covariance matrix. So, it is not surprising to see that many classifiers perform like the random classifier and have misclassification rates close to $2/3$. The MD classifier has the lowest misclassification rate. Among the rest, nonlinear SVM has a somewhat better result.

\begin{ex}\label{example7}
This is again a three class problem. The classes have normal distributions with the centres at the origin, differing only in their correlation structures. The corresponding covariance matrices have diagonal elements $1$, while all the off-diagonal elements are 0.1, 0.5 and 0.9 for the three classes, respectively.
\end{ex}
In this example with normal distributions differing in their covariance structures, as expected, QDA has the best performance. But, the MD classifier closely follows QDA and has much better performance than all other classifiers considered here.

\begin{ex}\label{example8}
We consider a classification problem between a  normal and a  Cauchy distribution with the same location ${\bf 0}_d$ and the same scatter matrix having all diagonal elements equal to $1$ and all off-diagonal elements equal to $0.1$.
\end{ex}
Since the Cauchy distribution does not have finite moments, in this example, we use the MCD estimates \citep{rousseeuw1999fast} of the scatter matrices (based on 75\% observations) to compute the empirical Mahalanobis distances. We use these estimates for LDA and QDA as well. Here also, the MD classifier outperforms its competitors and the difference becomes more prominent as the dimension increases.

\section{Classifier Based on Local Mahalanobis Distances\label{sec:LMD classifier}}

We have seen that the MD classifier performs very well when the underlying class distributions are elliptic. But, when the underlying densities are not elliptic, Mahalanobis distances may fail to extract substantial information about them. Especially, in the case of multimodal distributions, the local nature of the density functions may not be captured well by Mahalanobis distances. As a result, in a classification problem involving non-elliptic distributions, the proposed MD classifier may fail to have satisfactory performance (this is demonstrated later in the paper). To cope with such situations, we define a local version of the Mahalanobis distance. To this effect, first note that for any $\vec x\in \R^d$,
\[
(\vec x-\muvec)^\top \sigmat^{-1} (\vec x-\muvec) = E\left\{(\vec x-\vec X)^\top \sigmat^{-1}(\vec x-\vec X)\right\} - d,
\]
where the random vector $\vec X$ follows a $d$-dimensional distribution with mean $\muvec$ and covariance $\sigmat$. Similarly, for a collection of observations $\vec x_1,\ldots,\vec x_n$, if we define $\widehat\muvec = \bar{\vec x} = n^{-1} \sum_{i=1}^n \vec x_i$ and $\widehat\sigmat = \vec S = n^{-1} \sum_{i=1}^n (\vec x_i-\bar{\vec x})(\vec x_i-\bar{\vec x})^\top$ to be the moment-based estimators of location and scatter, then we have
\[
(\vec x-\widehat\muvec)^\top \widehat\sigmat^{-1}(\vec x-\widehat\muvec) = \frac{1}{n}\sum_{i=1}^{n}(\vec x-\vec x_i)^\top \widehat\sigmat^{-1}(\vec x-\vec x_i) - d.
\]
So, working with the squared Mahalanobis distance (respectively, its sample analog) is equivalent to working with its location shift version $E\{(\vec x-\vec X)^\top \sigmat^{-1}(\vec x-\vec X)\}$ (respectively, $n^{-1} \sum_{i=1}^{n}(\vec x-\vec x_i)^\top \widehat\sigmat^{-1}(\vec x-\vec x_{i})$). One can notice that additive functions of Mahalanobis distances are also additive functions of the location shift versions of squared Mahalanobis distances, so Theorem~\ref{thm:MD_GAM} remains valid for this version as well. Now, for a fixed point $\vec x$, $E\{(\vec x-\vec X)^\top\sigmat^{-1}(\vec x-\vec X)\}$ puts equal weight on all points $\vec X$ on the support of the distribution, irrespective of its relative position with respect to $\vec x$. Similarly, $n^{-1}\sum_{i=1}^{n}(\vec x-\vec x_i)^\top \widehat\sigmat^{-1}(\vec x-\vec x_i)$ puts equal weight on all observations $\vec x_1,\ldots,\vec x_n$. To capture the local nature of the underlying distribution, we need to use a weighted average, giving more weights to the points near $\vec x$. We achieve this by using the idea of kernel density estimation. Using a kernel function $K:\R^d \to \R_{+}$, which is the density function of a $d$-variate spherical distribution symmetric about the origin (i.e., $K(\vec t)=\Psi(\vec t^\top\vec t), \vec t \in \R^d$ for some decreasing function $\Psi: \R_{+} \to \R_{+}$), we define 
\begin{align*}
    \beta_{h}(\mathbf{x})
= & E\left[K\Big(\frac{\sigmat^{-1/2}(\mathbf{x}-\mathbf{X})}{h}\Big)(\mathbf{x}-\mathbf{X})^\top{\sigmat}^{-1}(\mathbf{x}-\mathbf{X})\right]\\
= & E\left[\Psi\Big(\frac{1}{h^2}{(\mathbf{x}-\mathbf{X})^{\top}{\sigmat}^{-1}(\mathbf{x}-\mathbf{X})}\Big)(\mathbf{x}-\mathbf{X})^\top{\sigmat}^{-1}(\mathbf{x}-\mathbf{X})\right]
\end{align*}
The behaviour of the function $\beta_h(\cdot)$ depends on the tuning parameter $h$. We call $h$ the ``localization parameter" as it determines the local weighting scheme in the calculation of $\beta_h(\cdot)$. The function $\beta_h$ exhibits different behaviors depending on the value of $h$. For large values of $h$, it behaves similar to the Mahalanobis distance. On the other hand, for small values of $h$, it behaves similar to the density of the underlying distribution. The formal result is given below.
\begin{lemma}\label{lemma:LMD_asymptotics}
(a) If $K$ is continuous at $\vec 0$, then $\beta_h(\vec x) \to K({\bf 0})\{(\vec x-\muvec)^\top\sigmat^{-1}(\vec x-\muvec)+d\}$ as $h \to \infty$.\\
(b) If $\int \|\vec z\|^3 K(\vec z)\,d\vec z < \infty$ and the gradient of the density $f$ of $\vec X$ is bounded (i.e., $\sup_{\vec x} \|\nabla f(\vec x)\|$ is finite), then $\beta_h(\vec x)/h^{d+2} \to |\sigmat|^{1/2} \kappa_2 f(\vec x)$ as $h \to 0$, where $\kappa_2 = \int \|\vec z\|^2\,K(\vec z)\,d\vec z$.
\end{lemma}
Based on this result, we define the local (squared) Mahalanobis distance (LMD) as
\begin{align}\label{eq:LMD_def}
\gamma^{h}(\vec x) = \begin{cases} 
    \beta_{h}(\vec x) & \text{if } h>1, \\
    \beta_{h}(\vec x)/{h^{d+2}} & \text{if } h\leq 1. \end{cases}
\end{align}
Lemma~\ref{lemma:LMD_asymptotics} shows that for a small value of $h$, $\gamma^h$ contains information about the underlying density function. Hence, it is quite meaningful to use them as features to construct a classification rule. In fact, the Bayes classifier is related to (the limiting value of) $\gamma^{h}$ via a GAM, as shown below.

\begin{thm}\label{thm:LMD_GAM}
In a $J$-class classification problem, if the underlying class distributions are absolutely continuous, then the posterior probabilities satisfy an additive logistic regression model given by
\[
p(j \mid \vec x) = \frac{\exp \Big(\zeta_j\big(\Gamma^{0}(\vec x)\big)\Big)}{1+\sum_{k=1}^{J-1}\exp\Big(\zeta_k\big(\Gamma^{0}(\vec x)\big)\Big)} \text{ for } j=1,\ldots,J-1 \text{ and } p(J \mid \vec x) = \frac{1}{1+\sum_{k=1}^{J-1}\exp\Big(\zeta_k\big(\Gamma^{0}(\vec x)\big)\Big)},
\]
where $\zeta_j\big(\Gamma^{0}(\vec x)\big) = \sum_{k=1}^{J} \zeta_{jk}\big(\gamma^{0}_k(\vec x)\big)$ is an additive function of $\Gamma^{0}(\vec x)=\big(\gamma^{0}_1(\vec x),\ldots,\gamma^{0}_J(\vec x)\big)$, and $\gamma^{0}_j(\vec x)=\lim_{h \downarrow 0} \gamma_j^{h}(\vec x)$ for $j=1,\ldots,J$.
\end{thm}
So, instead of MD, we can use LMD (with a suitable choice of $h$) as features and construct a classifier using GAM as before. In practice, we compute $\widehat\sigmat_1,\ldots,\widehat\sigmat_J$ from the data to get 
\begin{align*}
\widehat\beta_{h,j}(\vec x) &= \frac{1}{n_j}\sum_{i=1}^{n_j}\left\{K\bigg(\frac{\widehat\sigmat_j^{-1/2}(\vec x-\vec x_{ji})}{h}\bigg)(\vec x-\vec x_{ji})^\top{\widehat\sigmat}_j^{-1}(\vec x-\vec x_{ji})\right\}\\
&= \frac{1}{n_j}\sum_{i=1}^{n_j}\left\{\Psi\bigg(\frac{1}{h^2}{(\vec x-\vec x_{ji})^{\top}{\widehat\sigmat}_j^{-1}(\vec x-\vec x_{ji})}\bigg)(\vec x-\vec x_{ji})^\top{\widehat\sigmat}_j^{-1}(\vec x-\vec x_{ji})\right\}, \qquad j=1,\ldots,J.
\end{align*}
Estimates of the LMDs are obtained using the formula
\begin{align}\label{eq:LMD_empirical}
\widehat\gamma_j^{h}(\vec x) = \begin{cases}
    \widehat\beta_{h,j}(\vec x) & \text{if } h>1,\\
    \widehat\beta_{h,j}(\vec x)/{h^{d+2}} & \text{if } h \leq 1.\end{cases}
\end{align}
For the empirical version of LMD, we have a uniform convergence result similar to Lemma~\ref{lemma:MD_uniform_convergence} if the function $\Psi$ is bounded and Lipschitz continuous. The result is stated below.

\begin{lemma}\label{lemma:LMD_uniform_convergence}
Let $\vec X_1,\ldots,\vec X_n$ be independent copies of a $d$-dimensional random vector  $\vec X$, which has finite fourth moments. Let $\gamma^{h}(\cdot)$ be the LMD as defined in \eqref{eq:LMD_def} and $\widehat\gamma^{h}(\cdot)$ be its empirical version as defined in \eqref{eq:LMD_empirical}, based on $\vec X_1,\ldots,\vec X_n$. Suppose that the function $\Psi$ is bounded and Lipschitz continuous, and $\widehat\sigmat^{-1}$ is a $\sqrt{n}$-consistent estimator of $\sigmat^{-1}$, i.e., $\|\widehat\sigmat^{-1}\mkern-2mu-\sigmat^{-1}\|_{\rm F} = O_P(n^{-1/2})$. Then, for any fixed $h>0$ and any compact set $\mathcal C \subset \R^d$, we have $\sup_{\vec x \in \mathcal C} \big|\widehat\gamma^{h}(\vec x)-\gamma^{h}(\vec x)\big| = O_P(n^{-1/2})$.
\end{lemma}

Note that if we use the Gaussian kernel (i.e., $K(\vec t) = (2\pi)^{-d/2} e^{-(\vec t^\top\vec t)/2}$), then the corresponding function $\Psi(t)=(2\pi)^{-d/2}e^{-t^2/2}$ satisfies the properties mentioned in the above lemma. Also, although we have proved the lemma for a fixed $h$, from the proof it is easy to see that the same result holds for a sequence $\{h_n\}$ that is uniformly bounded away from $0$. That is, if $\{h_n\}$ is a sequence of real numbers satisfying $\inf_{n \ge 1} h_n \ge h_0$ for some $h_0>0$, then $\sup_{\vec x \in \mathcal C} |\widehat\gamma^{h_n}(\vec x) - \gamma^{h_n}(\vec x)| = O_P(n^{-1/2})$. So, for arbitrarily large $h$, the estimated LMD is close to its population counterpart, which in turn behaves like the shifted squared Mahalanobis distance. The case of small $h$, however, is more tricky. For a general sequence $\{h_n\}$, we have $\sup_{\vec x \in \mathcal C} |\widehat\gamma^{h_n}(\vec x)-\gamma^{h_n}(\vec x)| = O_P\big(n^{-1/2}h_n^{-(d+2)} \wedge n^{-1/2} h_n^{-(d+4)}\big)$. Thus, for $h_n \to 0$, we have the uniform convergence $\sup_{\vec x \in \mathcal C} |\widehat\gamma^{h_n}(\vec x) - \gamma^{h_n}(\vec x)| \overset{P}{\to} 0$ as $n \to \infty$ if $\sqrt{n}h_n^{d+4}\to \infty$. So, we cannot use an arbitrarily small value of $h$ in practice. It should decrease at an appropriate rate as the sample size increases.

To demonstrate the empirical behaviour of $\widehat\gamma^{h}(\cdot)$ for varying choices of $h$, we consider three simple examples, one involving elliptic and the other two involving non-elliptic bivariate distributions. The density contours of these distributions are shown on the first column in Figure~\ref{fig:choice_of_h_contour_plots}. In each case, we generate 1000 observations and compute the estimated MD contours (shown on the second column), and the estimated LMD contours for a small and a large value of $h$ (shown on the third and the {fourth} columns, respectively). For the first example, we generate observations from $N_2(\vec 0,\sigmat)$, where $\sigmat=0.5{\vec I}_2+0.5{\bf 1}_2{\bf 1}_2^\top$. Here, estimated MD contours almost coincide with the density contours. The LMD contours also well approximate the density contours both for large and small values of $h$, but the approximation is much better when a larger value is used. For the second example, we generate observations from the bivariate Laplace distribution with density $f(x_1,x_2) \propto \exp\{-(|x_1|+|x_2|)\}$. In this case, the estimated MD contours and the estimated LMD contours with large $h$ are very different from the density contours, but the LMD contours with small $h$ approximate the density contours well. We observe a similar phenomenon for the third example, involving the mixture normal distribution $0.5 N_2({\bf 1}_2,\sigmat)+0.5N_2(-{\bf 1}_2,\sigmat)$, where $\sigmat = \mathbf{I}_2 + 4 \vec 1_{2} \vec 1_{2}^\top$. Figure~\ref{fig:choice_of_h_contour_plots} clearly shows that for non-elliptic distributions, it is better to use LMD with small $h$ to capture more information about the density. On the other hand, for elliptic distributions, LMD with large $h$ and MD behave similarly, and they approximate the density contours better than {the} estimated LMD with small $h$, which has relatively higher stochastic variation.

\begin{figure}[h]
\centering
\setlength{\tabcolsep}{2pt}
\begin{tabular}{cccc}
\multicolumn{4}{c}{(a) Normal distribution} \\ [5pt]
Density & MD & LMD with $h=0.4$ & LMD with $h=5$ \\
\includegraphics[width=0.22\linewidth]{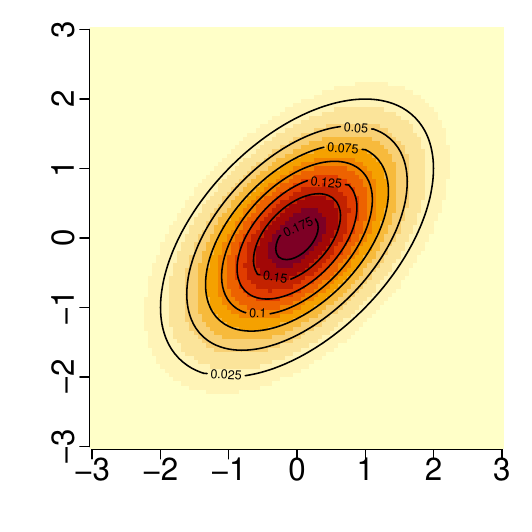} & \includegraphics[width=0.22\linewidth]{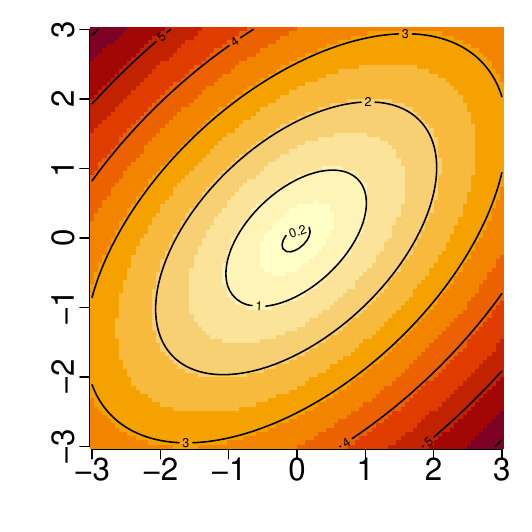} & \includegraphics[width=0.22\linewidth]{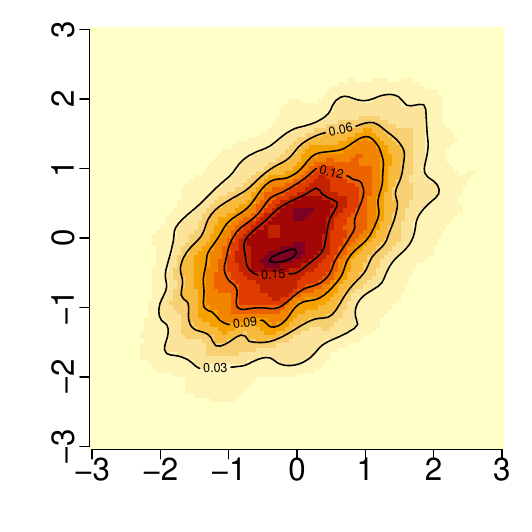} &  \includegraphics[width=0.22\linewidth]{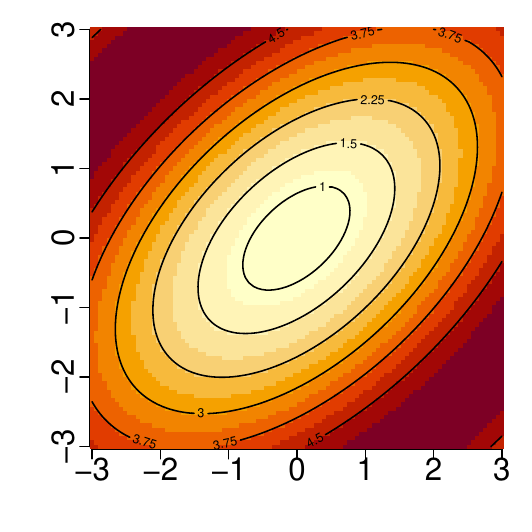} \\ [2pt]
\multicolumn{4}{c}{(b) Double Exponential distribution} \\ [5pt]
Density & MD & LMD with $h=0.3$ & LMD with $h=4$ \\
\includegraphics[width=0.22\linewidth]{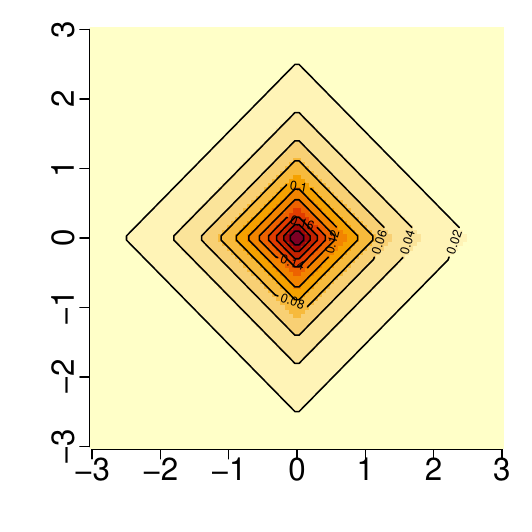} & \includegraphics[width=0.22\linewidth]{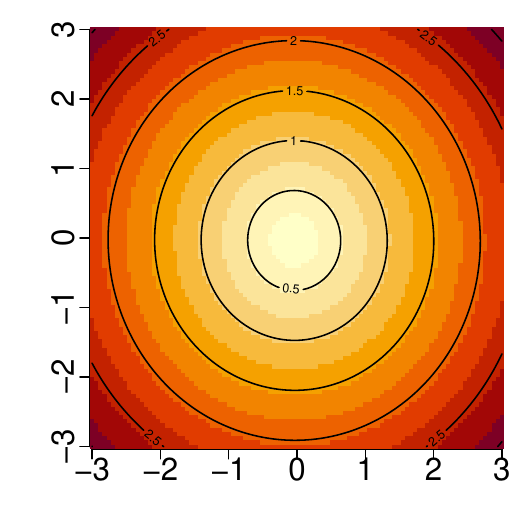} & \includegraphics[width=0.22\linewidth]{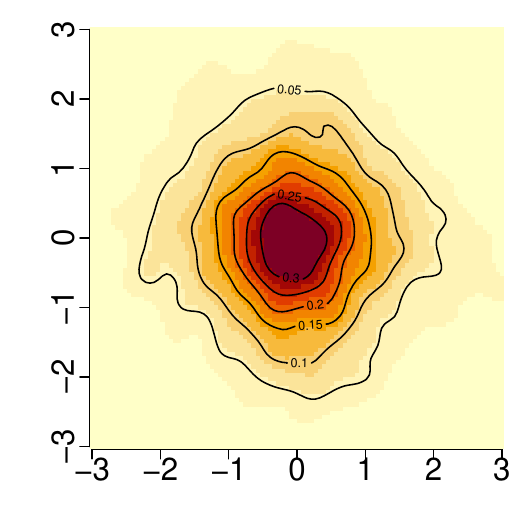} & \includegraphics[width=0.22\linewidth]{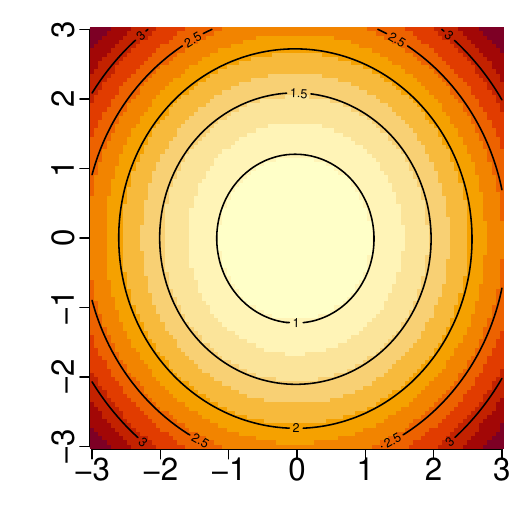} \\
\multicolumn{4}{c}{(c) Mixture Normal distribution} \\ [5pt]
Density & MD & LMD with $h=0.4$ & LMD with $h=5$ \\
\includegraphics[width=0.22\linewidth]{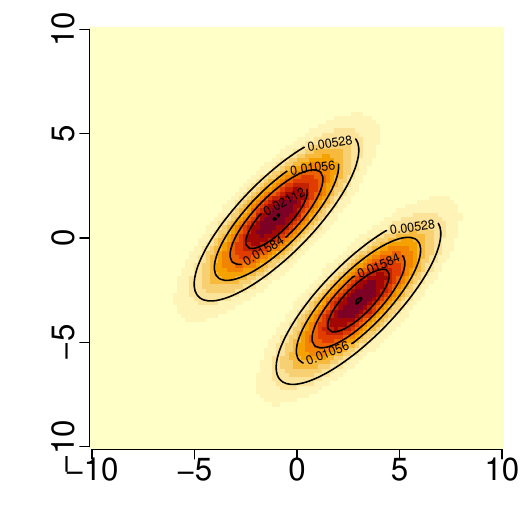} & \includegraphics[width=0.22\linewidth]{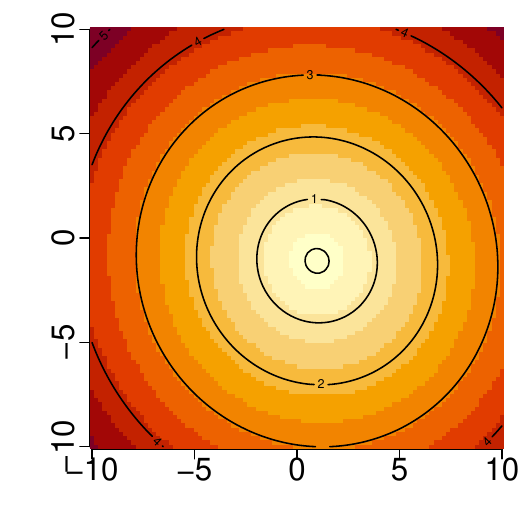} & \includegraphics[width=0.22\linewidth]{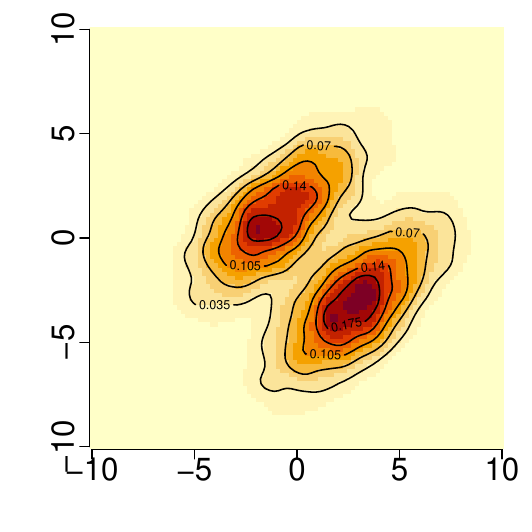} & \includegraphics[width=0.22\linewidth]{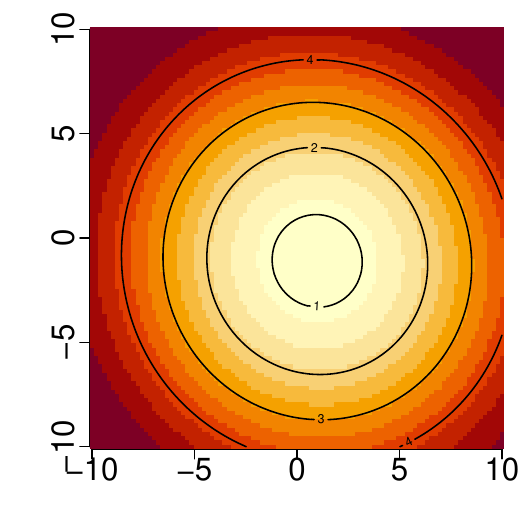} 
\end{tabular}
\caption{Contour plots of density, estimated MD, and estimated LMD with small and large values of $h$ in three examples involving elliptic (a) and non-elliptic (b and c) distributions.\label{fig:choice_of_h_contour_plots}}
\end{figure}

\begin{figure}[h]
\centering
\begin{tabular}{ccc}
\includegraphics[width=0.4\linewidth,height=2.25in]{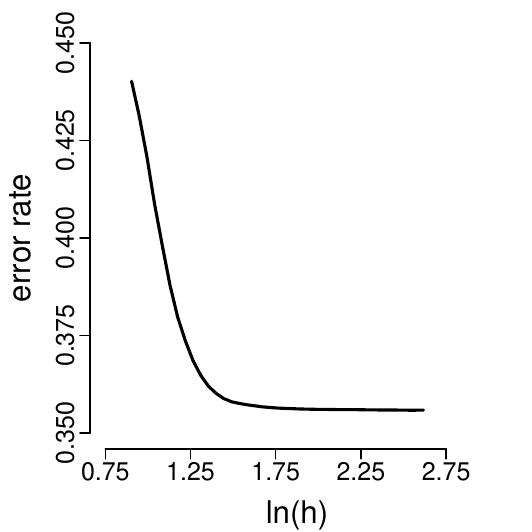} & & \includegraphics[width=0.4\linewidth,height=2.25in]{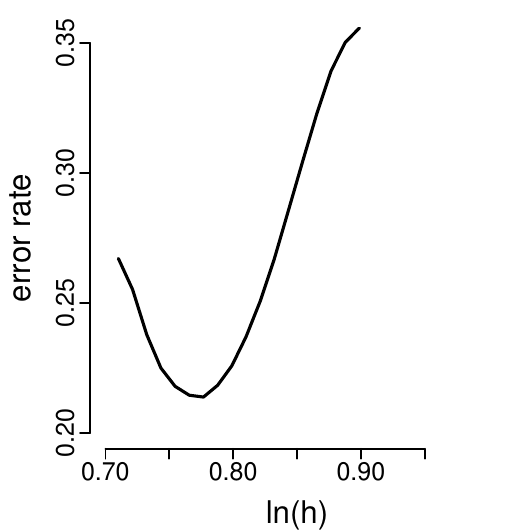}
\end{tabular}
 \caption{Average misclassification rates of the LMD classifier for different values of $h$ in Example~A (left) and Example~B (right) involving elliptical and non-elliptical distributions, respectively.\label{fig:LMD_error_varying_h}}
\end{figure}

A similar behavior can be observed for the classifier based on the local Mahalanobis distances. We formally define the classifier here. For a fixed $h>0$, we compute the empirical LMDs $\widehat\Gamma_h(\vec x_{ji}) = \big(\widehat\gamma^{h}_1(\vec x_{ji}),\ldots,\widehat\gamma^{h}_J(\vec x_{ji})\big)$ for $j=1,\ldots,J$ and $i=1,\ldots,n_j$. These are used as feature vectors, and a GAM with logistic link is fitted on these features to model the posterior probabilities of different classes (cf.\ Theorem~\ref{thm:LMD_GAM}). Finally, an observation is classified to the class with the highest posterior probability estimated from the fitted GAM. We refer to this classifier as the LMD classifier. To demonstrate the performance of the LMD classifier for different values of $h$, we consider two examples involving binary classification. In Example A, the class distributions $N_2(-0.3 {\bf 1}_2,{\bf I}_2)$ and $N_2(0.3 {\bf 1}_2,{\bf I}_2)$ differ only in their locations. In Example B, each of the competing classes is an equal mixture of two bivariate normal distributions. While one class is a mixture of $N_2\big((1,-1)^\top, \sigmat\big)$ and $N_2\big((-3,3)^{\top}, \sigmat\big)$, the other one is a mixture of $N_2\big((-1,1)^{\top}, \sigmat\big)$ and $N_2\big((3,-3)^{\top}, \sigmat\big)$, where $\sigmat={\bf I}_2+4{\bf 1}_2{\bf 1}_2^{\top}$.  For these two examples, we consider training and test samples of sizes 200 and 2000, respectively (equal number of observations from the two classes). The average misclassification rates of the LMD classifier for different choices of $h$ are shown in Figure~\ref{fig:LMD_error_varying_h}. In Example A, where the distributions are elliptic, the LMD classifiers with large values of $h$ perform much better than those based on small values of $h$. But, in Example B, where the distributions are non-elliptic, we observe a diametrically opposite picture. Here, the LMD classifiers based on large values of $h$ do not perform well, but those based on smaller values of $h$ (e.g., $h=2.16$) have excellent performance. In this example, while the MD classifier misclassifies more than 35\% observations, the LMD classifier with this choice of $h$ has a misclassification rate close to 22\%.

These examples clearly show that the performance of the LMD classifier depends heavily on the value of $h$, and it needs to be chosen appropriately. In this article, we use the bootstrap method for this purpose. More specifically, we fix a collection of $h$ values, say $\{h_1,\ldots,h_k\}$. For each of them, we construct the LMD classifier and compute the estimate of its average misclassification rate based on 100 bootstrap samples. Finally, the value of $h$ leading to the minimum average misclassification rate is selected for constructing the classifier. For selecting the collection $\{h_1,\ldots,h_k\}$, we use the following strategy. We compute the pairwise Mahalanobis distances among the observations in a class, separately for all the classes. One-third of the fifth percentile of these pairwise distances is taken to be the lower limit $h_1$. Since LMD behaves like a shifted version of the squared MD for large values of $h$, for the upper limit, we choose a value of $h$ for which $\mathrm{LMD}$ is close to $K(\vec 0)(\mathrm{MD}^2+d)$ (cf.\ Lemma~\ref{lemma:LMD_asymptotics}). In order to do this, we start with $h_1$ and keep increasing the value in a $\log$-scale taking $h_i=h_1 k_0^{i-1}$ for $i\ge 2$ and some $k_0>1$ until the correlation between $\big\{\widehat\gamma_j^{h}(\vec x_{ji}):j=1,\ldots,J;~i=1,\ldots,n_j\big\}$ and $\big\{\big(\widehat\delta_j(\vec x_{ji})\big)^2+d: j=1,\ldots,J;i=1,\ldots,n_j\big\}$ becomes higher than a threshold $r$ close to $1$. We stop at the first instance when the correlation is higher than $r$ and all the intermediate $h$ values (including the one at which we stop, say $h_k$) constitute the set $\{h_1,\ldots,h_k\}$. The performance of the LMD classifier in Examples~\ref{example1}--\ref{example8} with the bootstrapped selection of $h$ is shown in Table~\ref{tab:MD_vs_LMD}. We also show the results for the MD classifier for reference. It is clear from this table that the LMD classifier with bootstrapped choice of $h$ perform well and its misclassification rates are close to those of the MD classifier in Examples~\ref{example1}--\ref{example8}.

\begin{table}[h]
\setlength{\tabcolsep}{0.05in}
    \centering
    \small
    \begin{tabular}{c|c|cccccccccc}
    Classifier & $d$ & Ex 1 & Ex 2 & Ex 3 & Ex 4 & Ex 5 & Ex 6 & Ex 7 & Ex 8 \\ \hline
     & 2& 6.85 \se{.27}& 35.14 \se{.15}& 23.93 \se{.07} & 4.72 \se{.24}& 37.82 \se{.13}& 51.65 \se{.13}& 46.52 \se{.11}& 35.01 \se{.09}\\
    MD& 4& 7.75 \se{.18}& 29.59 \se{.10} & 15.27 \se{.06} & 5.18 \se{.14} &32.87 \se{.14} &53.07 \se{.13} & 32.32 \se{.12} &30.50 \se{.13}\\ 
     & 6 & 8.46 \se{.18}& 26.00 \se{.11} & 10.70 \se{.08} & 4.80 \se{.15} & 30.95 \se{.17} & 53.83 \se{.13} & 26.82 \se{.12} &28.30 \se{.29}\\ \hline
     
    &2 & 6.83 \se{.26} & 35.16 \se{.15}& 23.99 \se{.08} & 4.58 \se{.23} & 38.14 \se{.14} & 51.99 \se{.14} & 46.72 \se{.11}& 36.50 \se{.39} \\
    LMD &4 &7.74 \se{.18} & 29.58 \se{.09} & 15.31 \se{.07}  & 5.40 \se{.42}  &32.59 \se{.13} &52.83 \se{.13} & 32.44 \se{.12} &32.98 \se{.50}\\
    &6 & 8.42 \se{.18}& 25.96 \se{.11} & 10.65 \se{.09} & 4.55 \se{.12} & 30.57 \se{.17} & 53.24 \se{.14}&  27.29 \se{.26} & 32.46 \se{.72}\\ \hline
   \end{tabular}
    \caption{Average misclassification rates (in \%) of MD and LMD (with bootstrap choice of $h$) classifiers in Examples~\ref{example1}--\ref{example8}. The corresponding standard errors are reported within brackets in a smaller font.\label{tab:MD_vs_LMD}}
 \end{table}

In order to take care of non-elliptic distributions and complex class boundaries, several classifiers based on kernelized Mahalanobis distance (KMD) have been proposed in the literature \citep[e.g.,][]{mika1999fisher,ruiz2001nonlinear,chatpatanasiri2010new}. But unlike LMD, KMD does not behave like the Mahalanobis distance for large values of $h$. Moreover, for small values of $h$, it converges to a constant, not to a constant multiple of the density function. So, the Bayes classifier cannot be expressed as a function of KMDs. Instead of LMD, when we used GAM based on KMD for classification, it led to poor results. For instance, in Example~\ref{example1}, while the MD and LMD classifiers had misclassification rates around 7--9\%, the use of KMD yielded misclassification rates more than 25\%. So, in this article, we do not consider it for further study.

\subsection{Numerical Examples Involving Non-elliptic Distributions}
Here, we consider some classification problems involving non-elliptic distributions for further evaluation of the performances of MD and LMD classifiers. Misclassification rates of some popular state-of-the-art classifiers (those considered in Section~\ref{sec:MD_simulation}) are also reported for comparison. As in Section~\ref{sec:MD_simulation}, we generate 100 observations from each class to form the training set, and 1000 observations from each class (5000 for two-class problems) to form the test set. For each example, we consider three different choices of $d$ (2, 4 and 6) as before, and each experiment is repeated 100 times. The average test set misclassification rates and the corresponding standard errors of different classifiers over these 100 repetitions are reported in Table~\ref{tab:simulation_non-elliptic}.

\begin{ex}\label{example9}
We consider four multivariate Laplace distributions (the components are i.i.d.\ Laplace) with the same covariance matrix $0.75 {\bf I}_d$ but different locations. The locations of these four classes are taken to be $\vec 1_d$ , $-\vec 1_d$, $\vec a_d$ and $-\vec a_d$, respectively, where the components of $\vec a_d$ are $(-1)^i$ for $i=1,\ldots,d$.
\end{ex}
In this location problem, most of the classifiers barring the classification tree (CART) have satisfactory performance for all values of $d$. Among them, linear SVM has an edge for $d=2$ and LDA for $d=4$, but for $d=6$, the Random Forest classifier has the lowest misclassification rate. CART has higher misclassification rates for $d=4$ and $d=6$.
\begin{ex}\label{example10}
Here, we consider a classification problem between two multivariate exponential distributions ${\mathcal E}_d(1)$ and ${\mathcal E}_d(2)$, where ${\mathcal E}_d(\lambda)$ has $d$ i.i.d. component variables each following the exponential distribution with mean $\lambda$. 
\end{ex}
In this scale problem, the Bayes classifier has a linear class boundary. So, as expected, linear classifiers perform better than their competitors. Among them, logistic regression has the lowest misclassification rate. Nonlinear SVM, MD and LMD classifiers have competitive performance. They perform better than other nonlinear classifiers like $k$NN, KDA, CART and Random Forest.

\begin{table}[!t]
\setlength{\tabcolsep}{0.03in}
    \centering
    \small
    \begin{tabular}{c|c|cccccccccccccc}
       Data &$d$&Bayes & LDA & QDA & Logistic & GLM &  $k$NN & KDA & CART & Random & SVM & SVM & MD & LMD\\ 
        set  & &       &    & & Regression    & NET &     & & & Forest & Linear & RBF & \multicolumn{2}{c}{(Proposed)}\\ \hline
    & 2 & 24.57 & 24.84 & 25.23 & 24.91 & 24.91 &  25.24 & 25.06 & 25.69 & 28.70 & \best{24.83} & 25.58 & 25.64 & 25.73 \\
        & & \se{0.06} & \se{0.07} & \se{0.07} & \se{0.07} & \se{0.07} &  \se{0.08} & \se{0.07} & \se{0.11} & \se{0.12} & \se{0.07} & \se{0.08} & \se{0.07} & \se{0.07}\\ \cline{2-15}
              Ex 9 &4 &  15.58 & \best{15.96} & 16.52 & 16.27 & 16.31 &  16.22 & 16.08 & 20.59 & 17.53 & 16.29 & 16.66 & 17.17 & 17.20\\ 
    & &   \se{0.05} & \se{0.06} & \se{0.06} & \se{0.06} & \se{0.06}  & \se{0.07} & \se{0.07} & \se{0.19} & \se{0.08} & \se{0.05} & \se{0.08} & \se{0.07} & \se{0.06}\\ \cline{2-15}
          &6 & 7.74 & 10.19 & 11.07 & 10.80 & 10.73 &  10.30 & 10.01 & 19.20 & \best{9.46} & 10.85 & 10.47 & 11.98 & 12.04\\ 
          & &  \se{0.04} & \se{0.05} & \se{0.06} & \se{0.06} & \se{0.06} &  \se{0.05} & \se{0.05} & \se{0.18} & \se{0.08} & \se{0.06} & \se{0.07} & \se{0.07} & \se{0.07}\\ \hline  
      & 2 &   31.96 & 32.28 & 33.56 & \best{32.18} & 32.44 &  33.71 & 34.75 & 35.47 & 37.85 & 32.51 & 33.90 & 33.18 & 33.23\\
        & & \se{0.05} & \se{0.06} & \se{0.06} & \se{0.06} & \se{0.10} &  \se{0.19} & \se{0.23} & \se{0.25} & \se{0.15} & \se{0.08} & \se{0.20 }& \se{0.13} & \se{0.12}  \\ \cline{2-15}
              Ex 10 &4 & 24.89 & 25.69 & 27.15 &\best{25.45} & 25.74 & 27.99 &29.24 & 32.23 & 28.82 & 25.68 & 26.62 & 26.60 &26.76\\ 
    & & \se{0.05} & \se{0.06} & \se{0.07} & \se{0.06} & \se{0.09} &  \se{0.16} & \se{0.21} & \se{0.26} & \se{0.13} & \se{0.08} & \se{0.10} & \se{0.11} & \se{0.11}\\ \cline{2-15}
          &6  &20.22 & 21.43 & 22.88 & \best{20.99} & 21.30 & 25.13 & 26.26 & 31.13 & 23.68 & 21.18 & 22.57 & 22.59 & 22.74\\ 
          & & \se{0.05} & \se{0.07} & \se{0.06} & \se{0.07} & \se{0.09} &  \se{0.14} & \se{0.20} & \se{0.20} & \se{0.10} & \se{0.07} & \se{0.14} & \se{0.10} & \se{0.11} \\ \hline
         & 2 & 27.60 & 49.73 & \best{28.96} & 49.73 & 49.78 & 32.04 & 30.76 & 34.61 & 33.64 & 46.57 & 29.96 & 29.08 & 29.00 \\
        & &\se{0.04} & \se{0.10} & \se{0.07} & \se{0.10} & \se{0.09} & \se{0.16} & \se{0.14} & \se{0.25} & \se{0.13} & \se{0.41} & \se{0.22} & \se{0.10} & \se{0.09}\\ \cline{2-15}
              Ex 11 &4 &15.26 &49.27 & \best{17.37} & 49.28 & 49.35 & 23.87 & 22.64 & 31.74 & 21.77 & 45.47 & 17.61 & 17.52 & 17.50\\ 
    & & \se{0.03} &  \se{0.11} & \se{0.06} & \se{0.11} & \se{0.11} & \se{0.12} & \se{0.12} & \se{0.28} & \se{0.13} & \se{0.40} & \se{0.09} & \se{0.07} & \se{0.07}\\ \cline{2-15}
          &6 & 8.76 & 48.85 & \best{11.30} & 48.87 & 48.99 &  21.06 & 19.79 & 30.73 & 15.86 & 44.30 & 11.54 & 11.66 & 11.62\\ 
          & & \se{0.03}  & \se{0.12} & \se{0.05} & \se{0.11} & \se{0.11} & \se{0.13} & \se{0.14} & \se{0.23} & \se{0.16} & \se{0.39} & \se{0.09} & \se{0.08} & \se{0.08}\\ \hline            
       
                 & 2 & 0.00 & 49.99 & 50.80 & 49.99 & 49.92 & \best{0.00} & \best{0.00} & 1.83 & 1.79 & 59.49 & \best{0.00} & 44.22 & \best{0.00}\\
        & & \se{0.00} & \se{0.04} & \se{0.88} & \se{0.04} & \se{0.20} &  \se{0.00} & \se{0.00} & \se{0.07} & \se{0.04} & \se{0.39} & \se{0.00} & \se{0.69} & \se{0.00}\\ \cline{2-15}
              Ex 12 &4 &0.00 &50.02 & 50.53 & 50.02 & 50.00 & 0.26 & 0.18 & 1.78 & 0.50 & 57.65 & 0.15 & 47.69 & \best{0.00}\\ 
    & & \se{0.00} & \se{0.04} & \se{0.34} & \se{0.04} & \se{0.16} & \se{0.13} & \se{0.07} & \se{0.05} & \se{0.03} & \se{0.38} & \se{0.10} & \se{0.38} & \se{0.00}\\ \cline{2-15}
          & 6 & 0.00 & 50.03 & 50.17 & 50.03 & 49.89  & 10.95 & 17.49 & 1.85 & \best{0.14} & 55.73 & 15.85 & 49.40 & 5.96\\ 
          & & \se{0.00} & \se{0.03} & \se{0.15} & \se{0.03} & \se{0.26} &  \se{0.42} & \se{0.24} & \se{0.07} & \se{0.01} & \se{0.39} & \se{0.38} & \se{0.15} & \se{0.22} \\ 
     \hline

        & 2 & 0.00 & 50.44 & 50.16 & 50.43 & 50.22 &  \best{12.60} & 15.33 & 30.24 & 14.17 & 52.24 & 16.65 & 46.74 & 29.92\\
        & & \se{0.00} & \se{0.27} & \se{0.19} & \se{0.27} & \se{0.16} &  \se{0.17} & \se{0.19} & \se{0.46} & \se{0.20} & \se{0.21} & \se{0.23} & \se{0.32} & \se{0.29}\\ \cline{2-15}
              Ex 13 &4 &0.00 & 50.12 & 50.14 & 50.12 & 50.00 &  29.08 & \best{28.96} & 45.61 & 36.55 & 51.25 & 30.27 & 48.93 & 34.04\\ 
    & & \se{0.00} & \se{0.09} & \se{0.10} & \se{0.09} & \se{0.05}& \se{0.24}& \se{0.13} & \se{0.28} & \se{0.18} & \se{0.11} & \se{0.19} & \se{0.14} & \se{0.39}\\ \cline{2-15}
          &6 & 0.00 & 50.32 & 50.06 & 50.30 & 50.12 & 37.17 & 38.74 & 48.07 & 42.39 & 50.84 & \best{35.90} & 49.39 & 36.42\\ 
          & & \se{0.00} & \se{0.08} & \se{0.10} & \se{0.08} & \se{0.05} &  \se{0.33} & \se{0.31} & \se{0.23} & \se{0.17} & \se{0.10} & \se{0.23} & \se{0.12} & \se{0.33} \\ 
     \hline
     
 & 2 &0.00 & 49.83 & 43.43 & 49.83 & 49.59& 10.69 & 10.02 & 24.86 & 12.71 & 50.55 & \best{7.38} & 21.42 & 14.92 \\ 
        & &  \se{0.00} & \se{0.09} & \se{0.43} & \se{0.09} & \se{0.13} & \se{0.17} & \se{0.13} & \se{0.29} & \se{0.16} & \se{0.14} & \se{0.16} & \se{0.14} & \se{0.30}\\ \cline{2-15}
        Ex 14 & 4&  0.00 & 49.95 & 43.67 & 49.95 & 50.02 &  36.26 & 34.54 & 41.95 & 34.26 & 49.77 & 35.71 & 31.62 & \best{29.20}\\  
         & &\se{0.00} & \se{0.06} & \se{0.17} & \se{0.06} & \se{0.09} &  \se{0.28} & \se{0.22} & \se{0.27} & \se{0.16} & \se{0.08} & \se{0.22} & \se{0.18} & \se{0.18}\\ \cline{2-15}
           &6 & 0.00 & 49.92 & 43.86 & 49.92 & 49.90 & 43.87 & 42.64 & 43.75 & 36.50 & 49.66 & 38.53 & 33.87 & \best{32.17} \\  
         & & \se{0.00} & \se{0.05} & \se{0.14} & \se{0.05} & \se{0.06} & \se{0.19} & \se{0.24} & \se{0.21} & \se{0.16} & \se{0.06} & \se{0.23} & \se{0.16} & \se{0.18} \\ 
      \hline
       & 2 & 29.48 & 50.48 & 43.85 & 50.59 & 50.55 & \best{32.47} & 33.13 & 34.82 & 35.73 & 47.95 & 33.66 & 33.46 & 33.84\\
        & & \se{0.05} & \se{0.35} & \se{0.38} & \se{0.36} & \se{0.31} & \se{0.18} & \se{0.17} & \se{0.25} & \se{0.14} & \se{0.31} & \se{0.25} & \se{0.18} & \se{0.20}\\ \cline{2-15}
            Ex 15 &4 &20.24 & 50.83 & 40.98 & 51.00 & 50.09 & 25.03 & 26.80 & 32.80 & 26.15 & 46.36 & \best{23.98} & 26.42 & 27.23\\
    & & \se{0.04} & \se{0.5} & \se{0.46} & \se{0.50} & \se{0.41} & \se{0.16} & \se{0.22} & \se{0.25} & \se{0.13} & \se{0.31} & \se{0.24} & \se{0.3} & \se{0.43}\\ \cline{2-15}
          & 6 & 15.00 & 51.05 & 37.64 & 51.28 & 49.79 & 21.29 & 23.47 & 32.45 & 21.15 & 46.03 & \best{17.89} & 21.17 & 21.96\\ 
          & & \se{0.04} & \se{0.58} & \se{0.68} & \se{0.57} & \se{0.5} & \se{0.16} & \se{0.25} & \se{0.21} & \se{0.14} & \se{0.37} & \se{0.14} & \se{0.22} & \se{0.34}\\ 
     \hline
                      & 2 & 29.34 & 51.55 & 47.34 & 51.48 & 51.36 & 35.29 & 36.16 & 38.17 & 36.37 & 51.47 & 34.38 & \best{31.16}& 32.45\\
        & & \se{0.05} & \se{0.27} & \se{0.15} & \se{0.27} & \se{0.27} & \se{0.19} & \se{0.20} & \se{0.23} & \se{0.16} & \se{0.29} & \se{0.22} & \se{0.09} & \se{0.13}\\ \cline{2-15}
             Ex 16 &4 &20.24 & 52.08 & 45.71 & 51.95 & 51.87 & 31.25 & 31.08 & 38.34 & 29.65 & 53.18 & 27.36 & \best{23.04} & 24.40\\ 
    & & \se{0.04} & \se{0.29} & \se{0.19} & \se{0.29} & \se{0.29} & \se{0.19} & \se{0.15} & \se{0.27} & \se{0.15} & \se{0.25} & \se{0.29} & \se{0.11} & \se{0.15}\\ \cline{2-15}
          & 6 & 15.02 & 52.61 & 44.22 & 52.37 & 52.31 & 30.54 & 26.17 & 38.91 & 26.76 & 53.65 & 22.87 & \best{18.47} & 19.43\\ 
          & & \se{0.04} & \se{0.24} & \se{0.24} & \se{0.24} & \se{0.24} & \se{0.21} & \se{0.15} & \se{0.27} & \se{0.17} & \se{0.26} & \se{0.29} & \se{0.11} & \se{0.13}\\ 
     \hline
    \end{tabular}
    
  Boldface character signifies the best result in each case.
    \caption{Average misclassification rates (in \%) of different classifiers in Examples~\ref{example9}--\ref{examplenew2} for varying $d$. The corresponding standard errors are reported in the next line within brackets in a smaller font.\label{tab:simulation_non-elliptic}}
\vspace{-0.05in}
\end{table}

\vspace{0.1in}
Now, we consider some classification problems involving mixture distributions. In Examples 11 and 12, both of the competing classes are bimodal -- each of them is an equal mixture of two normal distributions. Example 13 deals with mixtures of several uniform distributions, where each of the component distributions is spherically symmetric. In Examples 14--16, we consider mixtures of non-elliptic distributions.

\begin{ex}\label{example11}
In this example, we consider two mixture normal distributions. The distribution corresponding to the first class is an equal mixture of $N_d(\vec 1_d,{\bf I}_d)$ and $N(-\vec 1_d,{\bf I}_d)$, whereas that for the second class is an equal mixture of $N_{d}(\vec a_d,4{\bf I}_d)$ and $N_{d}(-\vec a_d,4{\bf I}_d)$. Here, $\vec a_d$ is as defined in Example~\ref{example9}.
\end{ex}
Here, the Bayes class boundary is nonlinear. So, the nonlinear classifiers outperform the linear classifiers. Overall performances of QDA, nonlinear SVM, MD and LMD classifiers are much better than the other classifiers considered here.

\begin{ex}\label{example13}
For Class-1, all the component variables are i.i.d., and their distributions are equal mixtures of $N(1,0.01)$ and $N(-1,0.01)$. The distribution for Class-2 is a rotated version of the distribution for Class-1, where we rotate each successive pair of covariates by an angle of $45^\circ$.
\end{ex}

Here, the LMD classifier has significantly better performance than the MD classifier, irrespective of the dimension. Besides LMD, the overall performances of CART and Random Forest are also good. The performances of $k$NN and KDA are exceptionally good for $d=2$ and $4$.  But, their performance deteriorates as the dimension increases to $6$. The rest of the classifiers have very poor performance.

\begin{ex}\label{example12}
The two underlying populations are mixtures of three uniform distributions. Let $R_{(a,b)}(\vec z)$ be the region $\{\vec x \in \R^{d}: a \le \|\vec x-\vec z\| \le b\}$. Class-1 is a mixture of uniform distribuions on $R_{(0,1)}(-\vec c)$, $R_{(1,2)}(\vec c)$ and $R_{(2,3)}(-\vec c)$ with mixing proportions $0.25$, $0.5$ and $0.25$, respectively, where $\vec c=(5,0,\ldots,0)^\top$. Class-2 is obtained if $\vec c$ is replaced by $-\vec c$ in the description of Class-1. 
\end{ex}
Here also, the LMD classifier outperforms the MD classifier. For $d=2$, $k$NN has the lowest misclassification rate followed by Random Forest, KDA and nonlinear SVM. The LMD classifier has the next best misclassification rate. The performance of most of the classifiers, including $k$NN and Random Forest, deteriorates as the dimension increases. For $d=6$, while nonlinear SVM and the LMD classifier have misclassification rates close to 35\%, all other classifiers barring $k$NN and KDA misclassify more than 40\% observations.

\begin{ex}\label{example14}
For this example, we generate observations from the uniform distribution on the hypercube $[-2,2]^d$. An observation $\vec x=(x_1,\ldots,x_d)^\top$ is assigned to Class-$1$ if $1/2<\prod_{i=1}^{d}|x_i|<2$. Otherwise, it is assigned to Class-$2$.
\end{ex}
In this example, the LMD classifier performs much better than the MD classifier, especially for $d=2$. For $d=4$ and $d=6$. While most of the other classifiers perform poorly and have misclassification rates in excess of 40\%, MD and LMD classifiers have misclassification rates close to 30\%. Among the other classifiers, Random Forest and nonlinear SVM have relatively better performances.

\begin{ex}\label{examplenew1} In this example, Class 1 has the ${\mathcal E}_d(5)$ distribution, while Class-2 is an equal mixture of ${\mathcal E}_d(1)$ and ${\mathcal E}_d(10)$. Here ${\mathcal E}_d(\cdot)$ has the same meaning as in Example \ref{example10}.
\end{ex}

One can check that here the Bayes class boundary is nonlinear. So, as expected, all linear classifiers have poor performance. QDA also has higher misclassification rates.
In this example, all nonlinear classifiers have almost similar performance, with nonlinear SVM having an edge.

\begin{ex}\label{examplenew2} This example also deals with a mixture distribution. Here Class-1 has the ${\mathcal L}_d(5)$ distribution, while Class-2 is an equal mixture of ${\mathcal L}_d(1)$ and ${\mathcal L}_d(10)$, where ${\mathcal L}_d(\gamma)$ is a $d$-dimensional distribution with each component variables following i.i.d Laplace distribution with scale
parameter $\gamma$.
\end{ex}

In this example also, all linear classifiers and QDA perform poorly. Nonlinear classifiers have relatively better performance. Interestingly, MD and LMD classifiers have much lower misclassification rates than their competitors. Nonlinear SVM has the third-best performance.

\section{High Dimension, Low Sample Size Behaviour of MD and LMD Classifiers\label{sec:HDLSS}}

In this section, we consider the high dimension, low sample size (HDLSS) scenario, where the dimension of the data is much larger compared to the sample size. This type of data is frequently encountered nowadays in various areas including microarray gene expression studies, medical image analyses and spectral analysis in chemometrics. The characteristic property of this type of data is the scarcity of observations compared to the number of variables, which makes many existing methods unsuitable. For instance, popular parametric classifiers like LDA and QDA cannot be used in HDLSS situations due to the singularity of the sample covariance matrices. On the other hand, nonparametric classifiers like KDA and $k$NN often lead to poor performance in HDLSS situations, especially when the competing classes vary widely in their scales \citep[e.g.,][]{hall2005geometric, pal2016high, dutta2016some}. 

MD and LMD classifiers in their present form are also not usable in HDLSS situations. If the dimension of the data exceeds the sample size, we cannot use the sample covariance matrices for computing the empirical versions of the Mahalanobis distances or their localized versions since they are not invertible. Hence, for such data, we propose to modify the MD (respectively, LMD) classifier by using the identity matrix or the diagonals of the sample covariance matrices in the computation of the Mahalanobis distances (respectively, localized Mahalanobis distances). After computing the distances, the classifier is constructed as before by fitting a GAM with logistic link. Also, we use the bootstrap method to select the localization parameter $h$ for the LMD classifier. These modified classifiers have excellent empirical performance in Example~\ref{example1} with $d=500$. All of them correctly classify almost all observations, whereas many popular classifiers fail to achieve satisfactory performance (see Example~\ref{example20} in Table~\ref{tab:simulation_HDLSS}). Motivated by this, we now investigate the theoretical behaviour of these modified classifiers under the HDLSS asymptotic regime, where the sample sizes $n_1,\ldots,n_J$, with $\min\{n_1,\ldots,n_J\} \ge 2$, remain fixed and the dimension $d$ diverges to infinity. For our investigation, we make some assumptions which are given below. In the following, $\vec X=(X_1,\ldots,X_d)^\top$, $\vec Y=(Y_1,\ldots,Y_d)^\top$ and $\vec Z= (Z_1,\ldots,Z_d)^\top$ are generic random vectors. Also, we denote the mean and the covariance of Class-$j$ by $\muvec_j$ and $\sigmat_j$, respectively, for $j=1,\ldots, J$. 

\begin{enumerate}[({A}1)]
\item For all $1 \le j, j^\prime \le J$, and three independent random vectors $\vec X$ from Class-$j$ and $\vec Y, \Vec Z$ from  Class-$j^\prime$, we have $\Big|d^{-1} \sum_{i=1}^d (X_i-Y_i)(X_i-Z_i) - E\big\{d^{-1}\sum_{i=1}^d (X_i-Y_i)(X_i-Z_i)\big\}\Big| \stackrel{P}{\rightarrow} 0$ and $\Big|d^{-1} \sum_{i=1}^d (X_i-Y_i)^2 - E\big\{d^{-1}\sum_{i=1}^d (X_i-Y_i)^2\big\}\Big| \stackrel{P}{\rightarrow} 0$ as $d \rightarrow \infty$.

\item For all $j=1,\ldots,J$, there exists a constant $\sigma_j^2$ such that $d^{-1} \mathrm{trace}(\sigmat_j) \to \sigma_j^2$ as $d \to \infty$. Also, for every $1 \le j \neq j^\prime \le J$, there exists a constant $\nu_{jj^\prime}^2$ such that $d^{-1} \|\muvec_j-\muvec_{j^\prime}\|^2 \to \nu_{jj^\prime}^2$ as $d \to \infty$.
\end{enumerate}

These assumptions are pretty standard in the context of HDLSS asymptotics. \cite{hall2005geometric} considered the $d$-dimensional observations as time series truncated at time $d$ and studied the behaviour of pairwise distances as $d$ increases. They proved some distance convergence results under uniform boundedness of the fourth moments and $\rho$-mixing property of the time series. Assumption (A1) holds under those conditions. One can also assume some sufficient moment conditions like $\mathrm{Var}(\|\vec X-\muvec_j\|^2) = o(d^2)$, $\mathrm{Var}(\|\vec Y-\muvec_{j^\prime}\|^2)=o(d^2)$, $\mathrm{trace}(\sigmat^2_j)=o(d^2)$ and $\mathrm{trace}(\sigmat^2_{j^\prime})=o(d^2)$ for (A1) to hold \citep{banerjee2022high}. Some other relevant conditions can be found in \cite{ahn2007high,jung2009pca,aoshima2018two,sarkar2019perfect,yata2020geometric}. Under (A1) and (A2), we have high dimensional convergence of estimated Mahalanobis distances, as shown below. The following result is stated for the case when the identity matrix is used for computation of Mahalanobis distances from all competing classes. The case of diagonal covariance matrices is analogous, where we need to assume (A1) and (A2) for the standardized variables.

\begin{thm}\label{thm:MD_asymptotics_HDLSS}
{ Let $\vec X_{j1},\ldots,\vec X_{jn_j} \stackrel{iid}{\sim} F_j$ for $j=1,\ldots,J$ be independent collections of observations, and the distributions $F_1,\ldots,F_J$ satisfy (A1) and (A2). Define $\widehat\muvec_j = \overline{\vec X}_j = n_j^{-1}\sum_{i=1}^{n_j} \vec X_{ji}$ and suppose that the identity matrix ${\bf I}_d$ is used as $\widehat\sigmat_j$  to compute the empirical MD ${\widehat \delta}_j(\cdot)$ for $j=1,\ldots,J$. Then, we have the following results as $d \to \infty$.}
\begin{enumerate}[(a)]
    \item For any $j=1,\ldots,J$ and $i=1,\ldots,n_j$, $d^{-1}\,\widehat\Delta(\vec X_{ji}) \stackrel{P}{\rightarrow} \Theta_j=(\theta_{j1},\ldots,\theta_{jJ})$, where
    \[
    \theta_{jk}= \begin{cases} \Big(1-\frac{1}{n_j}\Big) \sigma_{j}^2 & \text{if } k=j,\\
    \nu_{jk}^2+\sigma_j^2+\frac{\sigma_k^2}{n_k} & \text{if } k \neq j.\end{cases}
    \]
    \item For an independent observation $\vec Z \sim F_i$, $d^{-1}\,\widehat\Delta(\vec Z) \stackrel{P}{\rightarrow} \Theta_i^{\ast}=(\theta^{\ast}_{i1},\ldots,\theta^{\ast}_{iJ})$, where
    \[
    \theta^{\ast}_{ik}=\begin{cases} \Big(1+\frac{1}{n_i}\Big) \sigma_{i}^2 & \text{if } k=i,\\
    \nu_{ik}^2+\sigma_i^2+\frac{\sigma_k^2}{n_k} & \text{if } k \neq i. \end{cases}
    \]
\end{enumerate}
\end{thm}

Part (b) of the above theorem shows that under (A1) and (A2), for a test observation $\vec Z$ from Class-$i$, the feature vector $\widehat\Delta(\vec Z)$ containing Mahalanobis distances tend to cluster (after scaling by $d$) around $\Theta_i^\ast$ when the dimension is large. Note that for $j \neq j^\prime$, $\Theta_j^\ast=\Theta_{j^\prime}^\ast$ if and only if $\sigma_j^2=\sigma_{j^\prime}^2$ and $\nu^2_{jj^\prime}=0$. So, if the distributions differ either in their locations or in their scales (i.e., for all $j \neq j^\prime$, either $\nu^2_{jj^\prime}>0$ or $\sigma_j^2 \neq \sigma_{j^\prime}^2$), then the limiting quantities $\Theta_1^\ast,\ldots,\Theta_J^\ast$ are all distinct. From part (a) of the theorem, it is also clear that if the sample sizes are not too small, then for a training sample observation $\vec X_{ji}$ from Class-$j$, $d^{-1}\widehat\Delta(\vec X_{ji})$ lies very close to $\Theta_j^\ast$. Hence, they also cluster around $J$ distinct points, one corresponding to each class, which are located close to $\Theta_1^\ast,\ldots,\Theta_J^\ast$. Naturally, the classifier based on GAM can separate out these $J$ points correctly. It partitions the observation space in such a way that these $J$ points belong to $J$ different regions formed by the partition. As a result, the training sample error becomes close to $0$ when $d$ is large. Since the test observations have similar convergence, they are also correctly classified by the MD classifier with high probability. So, the test set misclassification rate of the MD classifier is also often close to $0$ for large $d$. This is demonstrated in the left panel of Figure~\ref{fig:decision_boundary_normal_scale} for an experiment with observations from two competing normal populations $N_d(\vec 0_d,\vec I_d)$ and $N_d(\vec 0_d,1.5\vec I_d)$ differing only in their scales. We take $d=500$ and generate a  training (respectively, test) sample of size 200 (respectively, 2000) with an equal number of observations from each class. In Figure~\ref{fig:decision_boundary_normal_scale}, we can see that the estimated Mahalanobis distances in the training set and the test set are closely clustered. Moreover, the decision boundary from the fitted GAM perfectly separates the test data points. We observed a similar phenomenon in most of our numerical examples.

A result similar to Theorem~\ref{thm:MD_asymptotics_HDLSS} can be derived for LMD as well. The following theorem shows the high-dimensional behavior of LMD under (A1) and (A2) when the tuning parameter $h$ increases with the dimension at an appropriate rate. 

\begin{thm}\label{thm:LMD_asymptotics_HDLSS}
Let $\vec X_{j1},\ldots,\vec X_{jn_j} \stackrel{iid}{\sim} F_j$ for $j=1,\ldots,J$ be independent collections of observations, and the distributions $F_1,\ldots,F_J$ satisfy (A1) and (A2). {Suppose that the identity matrix ${\vec I}_d$ is used as $\widehat\sigmat_j$ to compute the empirical LMD ${\widehat \gamma}_j^h(\cdot)$ for $j=1,\ldots,J$.} Also, assume that $\Psi$ is continuous and $h$ increases with $d$ in such a way that $h^2/d \rightarrow C_0>0$ as $d \rightarrow \infty$. For $1 \le j,j^\prime \le J$, define $\theta_{jj^\prime}^{\circ} = \Psi\Big(\frac{\sigma_j^2+\sigma_{j^\prime}^2+\nu_{jj^\prime}^2}{C_0}\Big) (\sigma_j^2+\sigma_{j^\prime}^2+\nu_{jj^\prime}^2)$, where $\nu^2_{jj} = 0$. Then, we have the following results as $d \to \infty$.
\begin{enumerate}[(a)]
    \item For any $j=1,\ldots,J$ and $i=1,\ldots,n_j$, $d^{-1}\,\widehat\Gamma_h(\vec X_{ji}) \stackrel{P}{\rightarrow} \widetilde\Theta_j = (\widetilde\theta_{j1},\ldots,\widetilde\theta_{jJ})$, where
    \[
    \widetilde\theta_{jk} = \begin{cases} \Big(1-\frac{1}{n_j}\Big) \theta_{jj}^{0} & \text{if } k=j,\\ \theta_{jk}^{0} & \text{if } j \neq k.\end{cases}
    \]
    \item For an independent observation $\vec Z \sim F_i$, $d^{-1}\,\widehat\Gamma_h(\vec Z) \stackrel{P}{\rightarrow} \Theta_i^{\circ}=(\theta^{\circ}_{i1},\ldots,\theta^{\circ}_{iJ})$.
\end{enumerate}
\end{thm}

\begin{figure}[t]
\begin{center} 
\includegraphics[width=0.475\linewidth,,height=2.5in]{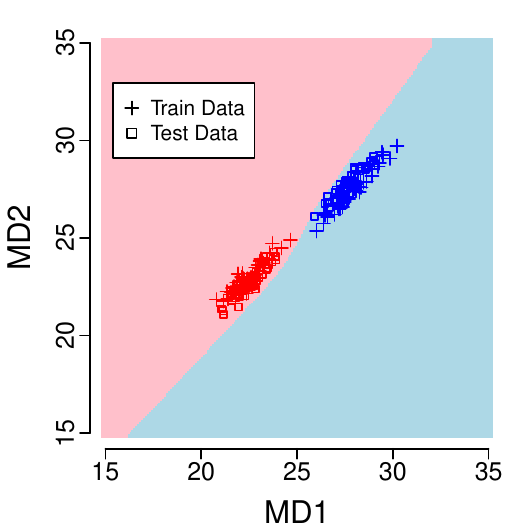}
\includegraphics[width=0.475\linewidth,height=2.5in]{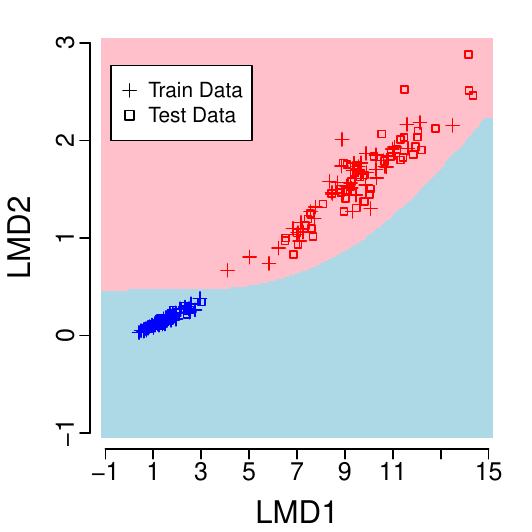}
\caption{Decision boundaries estimated by MD and LMD classifiers (with bootstrapped choice of $h$) along with scatter plots of empirical MD and LMD in a high-dimensional scale problem. { Here, MD1 and MD2 (respectively, LMD1 and LMD2) denote estimated Mahalanobis distances (respectively, local Mahalanobis distances) with respect to Class-1 and Class-2.}\label{fig:decision_boundary_normal_scale}}
\end{center}
\vspace{-0.15in}
\end{figure}

From the proof of the theorem, it can be seen that the result continues to hold for a data driven choice of the localization parameter $h$, as long as $h^2/d \overset{P}{\to} C_0$ as $d \to \infty$. If $h$ is chosen based on the median heuristic, it satisfies this condition. The condition is also satisfied if $h$ is chosen based on quantiles of pairwise distances. Theorem \ref{thm:LMD_asymptotics_HDLSS} shows that if the choice of the function $\Psi$ ensures that $\widetilde\Theta_1,\ldots,\widetilde\Theta_J$ are distinct, then the feature vectors of estimated LMDs obtained from the training data converge (after scaling by $d$) to $J$ distinct points, one for each class. As a result, the classifier based on GAM correctly classifies these $J$ points. If the training sample sizes are not too small, the limiting behaviour of the feature vectors for all the test set observations is almost the same as that of the training sample observations from the corresponding classes. Hence, the test cases are also correctly classified by the LMD classifier, and its error rate becomes close to $0$. This behavior of the LMD classifier can be seen in the right panel of Figure~\ref{fig:decision_boundary_normal_scale}, where we show the decision boundary of the LMD classifier and the scatter plots of the estimated LMDs for the training and the test samples.

Note that for the Gaussian kernel, the function $x\Psi(x/C_0)$ is monotonically increasing for $x<C_0$ and monotonically decreasing for $x>C_0$. So, in a two-class problem, if $d_{11}:=2\sigma_1^2$, $d_{12}:=\sigma_1^2+\sigma_2^2+\nu_{12}^2$ and $d_{22}:=2\sigma_2^2$ are on the same side of $C_0$, we have $\Theta_1^{\circ}=\Theta_2^{\circ}$ only when $d_{11}=d_{12}=d_{22}$, which in turn implies $\sigma_1^2=\sigma_2^2$ and $\nu_{12}^2=0$. One can always choose the tuning parameter $h$ (which determines $C_0$) in such a way that $d_{11},d_{12}$ and $d_{22}$ belong to either $(-\infty,C_0]$ or $[C_0, \infty)$. For instance, we can use the median heuristic on the observations from the two classes separately and use the smaller of the two as the tuning parameter. In that case, $\min\{d_{11},d_{22}\}/d$ plays the role of $C_0$ (note that $d_{12}$ cannot be smaller than both $d_{11}$ and $d_{22}$). A similar strategy works for classification problems involving more than two classes and also when the localization parameter is chosen based on the quantiles of the pairwise distances. However, this requirement is not necessary for $\Theta_j^{\circ} \neq \Theta_{j^\prime}^{\circ}$ to hold.

\subsection{Simulation Results}\label{sec:simulation_HDLSS}
To investigate the empirical behaviour of our proposed MD and LMD classifiers for high-dimensional data, here we consider some high-dimensional simulated examples. We take $d=500$, and in each case, the training set (respectively, test set) is formed by taking 100 (respectively, 1000) observations from each class. Each experiment is repeated 100 times to compute the average test set misclassification errors of different classifiers and the corresponding standard errors. These are reported in Table~\ref{tab:simulation_HDLSS}. For MD and LMD classifiers, we use both the identity matrix and the diagonals of the sample covariance matrices, and report the best result out of these two choices (separately for MD and LMD). The same strategy is adopted for LDA as well. For QDA, we use the diagonals of the sample covariance matrices.

\begin{ex}\label{example15}
We consider a location problem involving normal distributions $N_{500}(-0.2{\bf 1}_{500},\sigmat_{500})$ and $N_{500}(0.2{\bf 1}_{500},\sigmat_{500})$. Here, $\sigmat_{500}=((\rho^{|i-j|}))$ is an auto-correlation matrix with $\rho = 0.75$.
\end{ex}
As expected, in this location problem, LDA has the best performance, but the error rates of the other two linear classifiers, GLMNET and linear SVM, are  slightly higher. CART has poor performance in this high-dimensional example. Performances of all other classifiers are fairly satisfactory.

\begin{ex}\label{example16}
This is a scale problem with two multivariate normal distributions $N_{500}({\bf 0}_{500},\bm{\Gamma}_{500})$ and $N_{500}({\bf 0}_{500},1.5\,\bm{\Gamma}_{500})$, where $\bm{\Gamma}_{500} = ((\rho^{|i-j|}))$ with $\rho = 0.25$.
\end{ex}
In this problem, as expected, QDA performs well. On the other hand, all linear classifiers have misclassification rates close to 50\%. KDA and $k$NN also have poor performance. The reason behind the failure of $k$NN can be explained following \cite{hall2005geometric}. A similar argument can be given for KDA as well. CART also performs poorly, but the Random Forest has relatively better performance. Apart from QDA, nonlinear SVM, MD and LMD classifiers have excellent performance. Among them, the LMD classifier has an edge.

\medskip
The next three examples deal with binary classification involving mixture normal distributions. 
\begin{ex}\label{example17}
Class-$1$ is an equal mixture of $N_{500}(0.05{\bf 1}_{500},0.20{\bf I}_{500})$ and $N_{500}(-0.05{\bf 1}_{500}, 0.20{\bf I}_{500})$, while Class-$2$ is an equal mixture of $N_{500}(0.05{\bf a}_{500},0.25{\bf I}_{500})$ and $N_{500}(-0.05{\bf a}_{500},0.25{\bf I}_{500})$, where $\vec a_{500} = (-1,1,-1,1,\ldots,-1,1)^\top$ is as defined in Example~\ref{example9}.
\end{ex}
\begin{ex}\label{example18}
Class-$1$ is an equal mixture of $N_{500}(0.5{\bf 1}_{500},{\bf I}_{500})$ and $N_{500}(-0.5{\bf 1}_{500},4{\bf I}_{500})$, whereas Class-$2$ is an equal mixture of $N_{500}(0.5{\bf a}_{500},{\bf I}_{500})$ and $N_{500}(-0.5{\bf a}_{500}, 4{\bf I}_{500})$.
\end{ex}
In these two examples, the LMD classifier outperforms all its competitors. The MD classifier also has reasonably good performance in Example~\ref{example17}, but in Example~\ref{example18}, its misclassification rate is much higher. In these two examples, the competing classes do not satisfy assumptions (A1) and (A2), but each of the sub-classes satisfies these properties. So, depending on which of the two sub-classes an observation is coming from, the scaled version of $\widehat\Delta(\cdot)$ converges to two different points. It is not difficult to check that in the case of Example~\ref{example18}, these two points for one class are nearly the same as the corresponding points from the competing class (see the left panel of Figure~\ref{fig:decision_boundary_ex18}). This indistinguishability is the reason behind the poor performance of the MD classifier. For the LMD classifier, however, these four points (two from each class) are well separated (see the right panel of Figure~\ref{fig:decision_boundary_ex18}). As a result, the LMD classifier correctly classifies almost all observations.

\begin{figure}[t]
\begin{center} 
\includegraphics[width=0.475\linewidth,,height=2.5in]{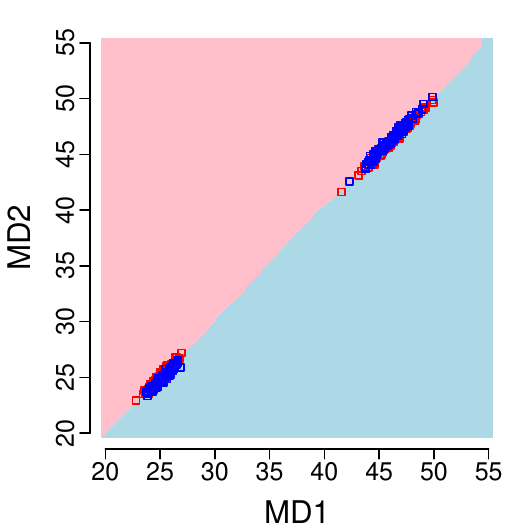}
\includegraphics[width=0.475\linewidth,height=2.5in]{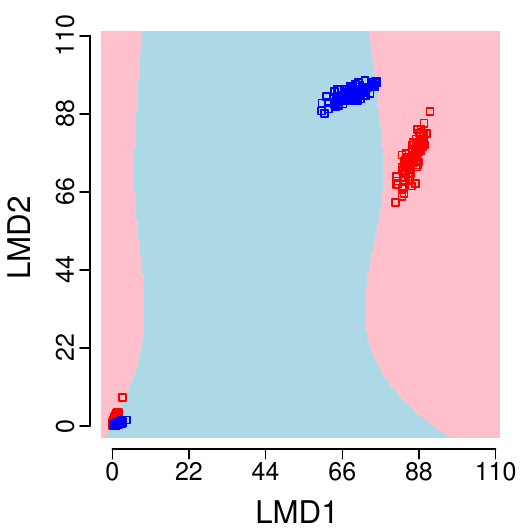}
\caption{Decision boundaries estimated by MD and LMD classifiers (with bootstrapped choice of $h$) along with scatter plots of empirical MD and LMD in Example~\ref{example18}. Here, MD1 and MD2 (respectively, LMD1 and LMD2) denote estimated Mahalanobis distances (respectively, local Mahalanobis distances) with respect to Class-1 and Class-2.\label{fig:decision_boundary_ex18}}
\end{center}
\vspace{-0.15in}
\end{figure}

\begin{table}[!b]
\setlength{\tabcolsep}{0.05in}
    \centering
    \small
    \begin{tabular}{c|ccccccccccccccc}
       Data & Bayes & LDA  & QDA & GLMNET & $k$NN & KDA & CART & Random & SVM & SVM & MD & LMD\\       
        set & & & & &  & & & Forest&  Linear& RBF & \multicolumn{2}{c}{(proposed)} \\ \hline
        
          Ex 17&  4.51 & \best{5.13} &  5.38 & 12.86 &  6.41 &  6.13 & 38.93 & 6.78  & 7.94 & 5.48 & 5.90  & 5.87  \\ 
           &  \se{0.08} & \se{0.09} & \se{0.09}& \se{0.16}&  \se{0.10}& \se{0.14}& \se{0.22} & \se{0.10}& \se{0.12}& \se{0.09} & \se{0.10} & \se{0.10} \\ \hline
          Ex 18&  0.00 & 48.48  & 0.91 & 49.28 & 50.00 & 50.00 & 44.83 & 12.98 & 49.11 & 1.74 &0.91  & \best{0.45} \\ 
          & \se{0.00}& \se{0.16} & \se{0.03 }& \se{0.16} & \se{0.00} & \se{0.00} & \se{0.25} & \se{0.16} & \se{0.15} &  \se{0.06} & \se{0.14} & \se{0.04} \\ \hline
          Ex 19&   0.00  & 49.70 & 50.00 & 49.94 & 49.90 & 49.99 & 48.26 & 33.41 & 50.08 & 19.61 & 13.50  & \best{5.17} \\ 
          &  \se{0.00}  & \se{0.16} & \se{0.00} & \se{0.17} & \se{0.02} & \se{0.01} & \se{0.18} & \se{0.17}& \se{0.15} & \se{0.18 }& \se{0.55} & \se{0.11} \\ \hline
          Ex 20&  0.00 & 48.70 & 49.33 & 43.87 &50.17 & 25.19 & 43.76& 22.54 & 42.34 & 28.69 & 37.86 &  \best{0.16} \\ 
          &  \se{0.00} & \se{0.40} &  \se{0.37} & \se{0.23}  & \se{0.15} & \se{0.12} & \se{0.29} & \se{0.40} & \se{0.15} & \se{0.53}& \se{1.50} & \se{0.03}\\ \hline
          Ex 21& 0.00 & 50.01 & 48.25 & 50.23  & \best{0.00} & \best{0.00}& 1.72 & \best{0.00} & 50.16 & \best{0.00} & 3.80 & \best{0.00}  \\ 
           & \se{0.00} & \se{0.11} & \se{0.78} & \se{0.23} & \se{0.00} & \se{0.00} & \se{0.05} & \se{0.00} & \se{0.11} & \se{0.00} & \se{0.73} & \se{0.00}\\ \hline
          Ex 22&  0.00 & 44.86  & 50.00 & 46.02 & 50.00 & 50.00 & 40.94 & 34.46 & 44.78 & 41.67 & \best{0.00}  & 0.14\\ 
          &  \se{0.00} & \se{0.13} & \se{0.00} & \se{0.17} &  \se{0.00} & \se{0.00} & \se{0.27}& \se{0.38} & \se{0.14} & \se{0.49} & \se{0.00} &  \se{0.08} \\ \hline
          Ex 23&4.28  &  48.64 & 22.15 & 47.34 & 50.00 & 14.26 & 40.69 & 17.37 & 45.85 & 15.82 & 17.42  & \best{6.84}   \\ 
           &  \se{0.0}  & \se{0.23}& \se{1.36} & \se{0.18}  & \se{0.00} & \se{0.15} & \se{0.28} & \se{0.12}& \se{0.11}& \se{0.13}  & \se{0.67}& \se{0.49} \\ \hline
           Ex 24&  0.00 & 50.08 & 50.00 &49.52 & 49.94 & 50.02&  50.04 &50.01 & 48.70 & 49.31 & 48.94  & \best{1.13} \\ 
          &   \se{0.00} & \se{0.15} & \se{0.00} & \se{0.16} & \se{0.03} & \se{0.03} & \se{0.13} & \se{0.17} & \se{0.15} & \se{0.17} & \se{0.42} & \se{0.68} \\ \hline
    \end{tabular}

    Boldface character signifies the best result in each case.
    \caption{Average misclassification rates (in \%) of different classifiers in Examples~\ref{example15}--\ref{example22}. The standard errors are reported in the next line within brackets in a smaller font.\label{tab:simulation_HDLSS}}
\end{table}

\begin{ex}\label{example19}
Both classes are equal mixtures of four normal distributions, each with covariance matrix $0.01\vec I_d$. For Class-$1$, the locations of the four sub-populations are $\vec a_d$, $-\vec a_d$, $\vec 1_d$ and $-\vec 1_d$. For Class-$2$, we take a rotation of the distribution for Class-$1$, where we rotate each successive pair of covariates by $45^\circ$ as done in Example~\ref{example14}.
\end{ex}
Here, the LMD classifier has zero misclassification rate along with $k$NN, KDA, Random Forest and Nonlinear SVM. CART and the MD classifier also have misclassification rates less than $5\%$. The rest of the classifiers have misclassification rates close to $50\%$.

\begin{figure}[b!]
\begin{center} 
\includegraphics[width=0.475\linewidth,,height=2.5in]{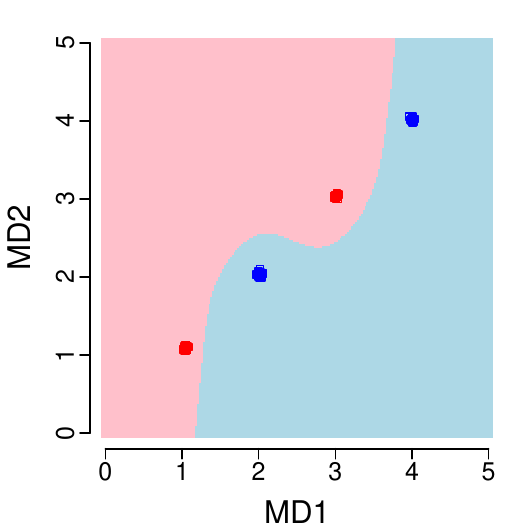}
\includegraphics[width=0.475\linewidth,height=2.5in]{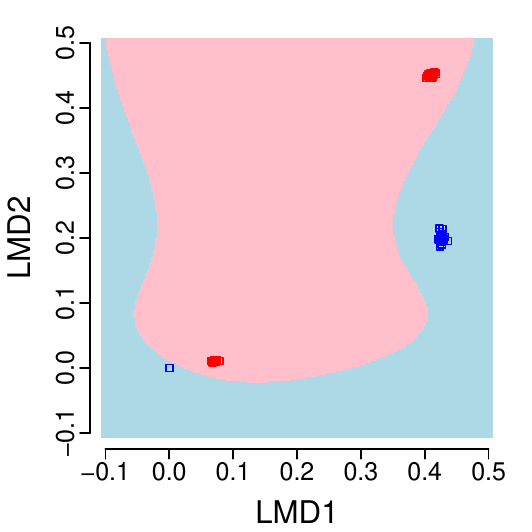}
 \caption{Decision boundaries estimated by MD and LMD classifiers (with bootstrapped choice of $h$) along with scatter plots of empirical MD and LMD in Example~\ref{example20}. Here, MD1 and MD2 (respectively, LMD1 and LMD2) denote estimated Mahalanobis distances (respectively, local Mahalanobis distances) with respect to Class-1 and Class-2.\label{fig:decision_boundary_ex20}}
\end{center}
\vspace{-0.15in}
\end{figure}

\begin{ex}\label{example20}
Class-1 is an equal mixture of $U_{500}(0,1)$ and $U_{500}(2,3)$, while Class-2 is an equal mixture of $U_{500}(1,2)$ and $U_{500}(3,4)$. Here, $U_{500}(a,b)$ is the uniform distribution on $\{\vec x \in \R^{500}: a \le \|\vec x\| \le b\}$, as defined in Example~\ref{example1}.
\end{ex}
This is a high-dimensional version of Example~\ref{example1}. Here, one can show that the sub-classes satisfy (A1) and (A2) \citep{sarkar2019perfect}. This is also quite evident from Figure~\ref{fig:decision_boundary_ex20}, where we plot the decision boundaries of the MD and LMD classifiers along with the scatter plots of the estimated MD and LMD values for the two classes. In this example, both MD and LMD classifiers have excellent performance. Their misclassification rates are very close to zero, while all the other classifiers have much higher misclassification rates.

\begin{ex}\label{example21}
We consider a classification problem between $N_{500}({\bf 0}_{500},3{\bf I}_{500})$ and the $500$-dimensional standard $t$ distribution with $3$ degrees of freedom.
\end{ex}
This is the high-dimensional version of Example~\ref{example5}, where the underlying distributions have the same mean and covariance, but differ in their shapes. Here, the assumption (A1) does not hold, and we analyze this dataset to investigate the performance of the MD and LMD classifiers under departure from the stipulated assumptions. Table~\ref{tab:simulation_HDLSS} shows that even in this example, the LMD classifier outperforms its competitors. Among the rest, KDA, nonlinear SVM, Random Forest, and the MD classifier have relatively better performance.

\begin{ex}\label{example22}
Class-1 is a mixture of uniform distributions on $R_{(0,1)}(-\vec c)$, $R_{(1,2)}(\vec c)$ and $R_{(2,3)}(-\vec c)$ with mixing proportions $0.25$, $0.5$ and $0.25$, respectively. Class-2 is a mixture of uniform distributions on $R_{(0,1)}(\vec c)$, $R_{(1,2)}(-\vec c)$ and $R_{(2,3)}(\vec c)$ with mixing proportions $0.25$, $0.5$ and $0.25$, respectively. Here, $R_{(a,b)}(\vec z) = \{\vec x \in \R^{500}: a \le \|\vec x - \vec z\| \le b\}$ and $\vec c = (5,0,\ldots,0)^\top$ are as defined in Example~\ref{example13}.
\end{ex}
This example is the high-dimensional version of Example~\ref{example13}. Here,  the LMD classifier has excellent performance, while all other classifiers including the MD classifier misclassify almost half of the observations. The reason behind this diametrically opposite behaviour of MD and LMD classifiers is the same as that in Example~\ref{example18}; indistinguishability (respectively, distinguishability) of the two populations in terms of the Mahalanobis distances (respectively, the localized Mahalanobis distances). This is also visible in Figure~\ref{fig:decision_boundary_ex22}.

\begin{figure}[t!]
\begin{center} 
\includegraphics[width=0.475\linewidth,,height=2.35in]{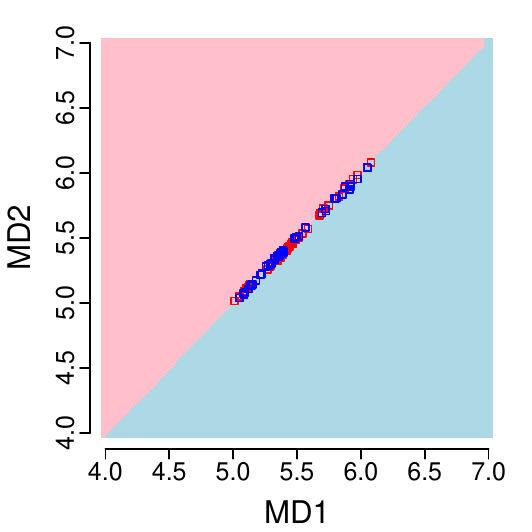}
\includegraphics[width=0.475\linewidth,height=2.35in]{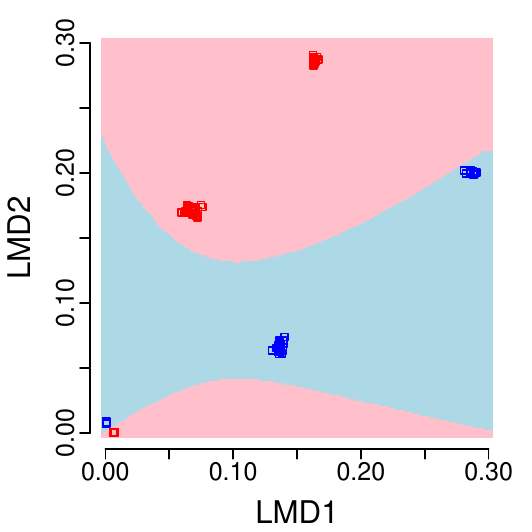}
 \caption{Decision boundaries estimated by MD and LMD classifiers (with bootstrapped choice of $h$) along with scatter plots of empirical MD and LMD in Example~\ref{example22}. {Here, MD1 and MD2 (respectively, LMD1 and LMD2) denote estimated Mahalanobis distances (respectively, local Mahalanobis distances) with respect to Class-1 and Class-2.}\label{fig:decision_boundary_ex22}}
\end{center}

\vspace{-0.2in}
\end{figure}

\section{Analysis of Benchmark Datasets\label{sec:real data}}
We analyze 18 benchmark datasets for further comparison of the performance of the proposed classifiers with other popular classifiers. Brief descriptions of these datasets are given in Table~\ref{tab:datasets}. ``Synthetic" and ``Biomed" datasets are taken from the CMU data archive (\url{http://lib.stat.cmu.edu/datasets/}). The Biomed dataset contains some observations with missing values, which we exclude from our analyses. ``Blood Transfusion" and ``Pima Indian Diabetes" datasets are taken from Kaggle (\url{https://www.kaggle.com/datasets}). The ``Colon Cancer" dataset can be obtained from the \texttt{R} package \texttt{rda}. ``Iris", ``Vehicle", ``Statlog (Landsat Satellite)" and ``Control Chart" datasets are taken from the UCI Machine Learning Repository (\url{https://archive.ics.uci.edu/datasets}). The Control Chart dataset contains six classes, out of which we consider only two classes, viz., `Normal' and `Cyclic' for our analysis. The rest of the datasets are taken from the UCR Time Series Classification Archive (\url{https://www.cs.ucr.edu/~eamonn/time_series_data_2018/}). Performances of different classifiers on these datasets are reported in Table~\ref{tab:benchmark_results}. The overall performance of GLMNET on these benchmark datasets is much better than that of logistic regression. Moreover, for some high-dimensional datasets, we received several warning messages while running the \texttt{R} codes for logistic regression. So, the error rates of logistic regression are not reported here.

\begin{table}[!b]
    \centering
    \small
    \begin{tabular}{|ccccc |c| ccccc|} \cline{1-5} \cline{7-11}
     Dataset & $d$ & $J$& \multicolumn{2}{c|}{Sample size} &  & dataset & $d$ & $J$& \multicolumn{2}{c|}{Sample size} \\
     &  & & Train & Test & & &  & & Train & Test\\ \cline{1-5} \cline{7-11}

         Synthetic $\ast$ & 2 & 2 & 250 & 1000 & &
        {Sony Robot Surface 1 $\ast$} & 70 & 2 & 20 & 601 \\ \cline{1-5}\cline{7-11}

          Blood Transfusion & 3 & 2 & 375 & 373 & &
        {Trace $\ast$} & 275 & 4 & 100 & 100\\ \cline{1-5} \cline{7-11}
        
          Iris & 4 & 3 & 75 & 75 & &   
        {Ham $\ast$} & 431 & 2 & 109 & 105\\ \cline{1-5} \cline{7-11}
         
        Biomed & 4 & 2 & 100 & 94 & &
        {Meat $\ast$} & 448 & 3 & 60 & 60\\ \cline{1-5} \cline{7-11}
        
        Pima Indian Diabetes & 8 & 2 & 385 & 383 & &
        {Earthquakes $\ast$} & 512 & 2 & 322 & 139\\ \cline{1-5} \cline{7-11}

         Vehicle & 18 & 4 & 425 & 421 & &
      Lightning 2 $\ast$ & 637 & 2 & 60 & 61   \\ \cline{1-5} \cline{7-11}

        {Italy Power Demand $\ast$} & 24 & 2 & 67 & 1029 & & 
        {Worms Two-class $\ast$} & 900 & 2 & 181 & 77 
        \\ \cline{1-5} \cline{7-11}

         Statlog $\ast$ 
         & 36 & 6 & 4435 & 2000      & &
       Colon Cancer & 2000 & 2 & 30 &32   \\ \cline{1-5} \cline{7-11}

        Control Chart & 60 & 2 & 100 & 100 & &
        {House Twenty $\ast$} & 2000 & 2 & 40 & 119 \\ \cline{1-5} \cline{7-11}
    \end{tabular}\\
    $\ast$ Datasets having specific training and test sets.
    
    \caption{Brief descriptions of the benchmark datasets.\label{tab:datasets}}
\end{table}

Synthetic, Statlog and all datasets from the UCR Time Series Archive have specific training and test samples. For these datasets, we report the test set misclassification rates of different classifiers. If a classifier has a misclassification rate $\varepsilon$, its standard error is computed as $\sqrt{\varepsilon(1-\varepsilon)/n_t}$, where $n_t$ is the size of the test set. For all other datasets, we form training and test samples by randomly partitioning the observations from each competing class into two parts of almost equal sizes. This random partitioning is done 100 times, and we report the average test set misclassification rates of different classifiers over these 100 random partitions along with their corresponding standard errors. Some of these datasets (see Table \ref{tab:datasets}) have dimensions larger than the corresponding training sample sizes. In those cases, both the identity matrix and the diagonals of the sample covariance matrices are used (as in Section~\ref{sec:HDLSS}) for constructing the MD classifier, and finally the one with a lower error rate is reported here. The same strategy is adopted for LDA and the LMD classifier as well. As before, QDA is implemented using the diagonals of the sample covariance matrices corresponding to the different classes. This strategy is also used for Italy Power Demand data, where due to collinearity among the feature variables, the pooled sample dispersion matrix turns out to be singular. In the case of Control Chart Data, the training sample sizes from each class is smaller than the data dimension. So, there also we use this strategy for QDA, MD and LMD classifiers. 

In \textbf{Synthetic} data, each class is a mixture of two distributions \citep{ripley2007pattern}, and the optimal class boundary is nonlinear. So, as expected, nonlinear classifiers (barring CART) perform better than linear classifiers. KDA and $k$NN have the lowest misclassification rate (8.4\%) followed by nonlinear SVM (9.2\%). QDA, MD and LMD classifiers have the same misclassification rate of 10.2\%, which is lower than the error rates of the remaining classifiers.

The \textbf{Blood Transfusion} data are collected from the Blood Transfusion Service Center in Hsin-Chu City, Taiwan. This dataset contains information on 748 blood donors and the aim is to predict whether the donor donated blood in March of 2007 \citep{yeh2009knowledge}. In this dataset, two features `frequency' (total number of blood donations)  and `monetary' (total blood donated in cc) have a perfect linear relation. So, we exclude `frequency' from our analysis. Here, all classifiers have almost similar performance with nonlinear SVM, MD and LMD classifiers having an edge.

\textbf{Iris} data contains measurements of sepal length, sepal width, petal length and petal width for three different species of Iris flowers: `Setosa', `Virginica' and `Versicolor'. This dataset was first analyzed by \cite{fisher1936use}, and it is known that the underlying distributions are almost normal. Unsurprisingly, LDA and QDA perform better than the other classifiers in this example. Linear SVM has the next best performance followed by the MD and LMD classifiers. All other classifiers except CART have similar misclassification rates. The error rate of CART is slightly higher.

\textbf{Biomed} dataset was created by \cite{cox1982exposition}. It contains information on 4 different measurements for 209 blood samples (134 `normals’ and 75 `carriers’ of a disease). Out of these 209 observations, 15 have missing values. We remove them and carry out our analysis using the remaining 194 observations (127 `normals’ and 67 `carriers’). As suggested by the diagnostic plots of \cite{li1997some}, the distributions of the two classes are nearly elliptic in this example. So, the MD classifier performs slightly better than the LMD classifier. Here, linear SVM and GLMNET have lower misclassification rates than others. QDA, MD and LMD classifiers also perform well, and they perform better than the rest of the competitors.

The \textbf{Pima Indian Diabetes} dataset was originally obtained from the US National Institute of Diabetes and Digestive and Kidney Diseases \citep{smith1988using}. It contains 8 measurements on 768 Pima Indian women (268 `diabetic' and 500 `non-diabetic') of age 21 years or more residing near Phoenix, Arizona, USA. In this example, all linear classifiers have lower error rates, but nonlinear classifiers also have satisfactory performance. Among the nonlinear methods, KDA and $k$NN have slightly higher error rates. Error rates of other nonlinear classifiers are almost similar.

In the \textbf{Vehicle} dataset, we have $18$ features extracted from the silhouettes of four different types of vehicles. The goal is to find the vehicle type from the extracted features. Here, the data distributions are nearly elliptic, which can be verified from the diagnostic plots of \cite{li1997some}. So, the MD classifier has slightly better performance than the LMD classifier. QDA has the lowest error rate in this example, followed by the MD and LMD classifiers. Except for $k$NN, KDA and CART, the other classifiers have reasonable performances.

The \textbf{Italy Power Demand} dataset deals with electrical power demands in Italy during Cold (October to March) and Warm (April to September) periods of a year. Here, each observation has 24 features, representing hourly power consumption over a day. The aim is to distinguish between the two periods based on the hourly power consumption patterns. In this dataset, all classifiers barring LDA and QDA perform well. Among them, CART and GLMNET have lower misclassification rates. The LMD classifier performs slightly better than the MD classifier.

The \textbf{Statlog} (Landsat Satellite) dataset contains information on multi-spectral values of pixels in $3 \times 3$ neighbourhoods in a satellite image. This is a 6-class classification problem where we want to classify the central pixel into one of six classes, each representing one of the six types of soil. In this example, nonlinear SVM has the lowest misclassification rate followed by Random Forest. KDA and $k$NN also perform well. The misclassification rates of MD and LMD classifiers are lower than the remaining classifiers. 

The \textbf{Control Chart} dataset contains examples of control charts synthetically generated using the process described by \cite{alcock1999time}. For our analysis, we consider $200$ observations belonging to the two classes `Cyclic' and `Normal'. Here, each observation is a time series observed at $60$ time points. In this example, all classifiers have very good performance and many of them including MD and LMD correctly classify all observations. 

The \textbf{Sony AIBO Robot Surface 1} dataset contains information from roll, pitch, and yaw accelerometers of a robot named Sony1, focusing solely on the X-axis. The primary objective is to identify the type of surface being traversed, specifically differentiating between the two classes: `cement' and `carpet'. In this dataset, LDA, QDA and the LMD classifier outperform their competitors. While they have an error rate of only 6.49\%, the MD classifier has the next best performance with an error rate of 16.81\%.

The \textbf{Trace} dataset is a synthetic dataset created by \cite{roverso2000multivariate} to simulate instrumentation failures in a nuclear power plant. This dataset includes four classes, which is a subset of a larger dataset with 16 classes. Each class consists of 50 instances, each of length 275. Here, MD and LMD classifiers perform exceptionally well with a misclassification rate of 8\%. Nonlinear SVM has the next best performance with 16\% misclassification.

The \textbf{Ham} dataset contains observations from two classes: `Spanish dry-cured hams' and `French dry-cured hams'. It has 109 training instances and 105 testing instances, obtained from multiple measurements of each ham sample. In particular, each of the 19 Spanish and 18 French ham samples are measured multiple times, resulting in a total of 214 observations. Each observation consists of 431 features, which represent different spectral properties of the ham samples. In this example, the MD classifier has the lowest misclassification rate followed by LMD. LDA, QDA, Nonlinear SVM and Random Forest perform better than the other competitors.

The \textbf{Meat} dataset consists of three classes: `chicken', `pork', and `turkey'. Each instance in the dataset has features obtained using Fourier transform infrared (FTIR) spectroscopy with attenuated total reflectance (ATR) sampling. The classification problem is to distinguish between chicken, pork, and turkey based on their spectroscopic data. For this dataset, MD and LMD classifiers have the lowest misclassification rate of 3.33\%, followed by QDA and GLMNET with 5\% misclassification. LDA and linear SVM have the next lowest misclassification rate of 6.67\%. Nonlinear SVM and CART have higher misclassification rates in this dataset.

\textbf{Earthquakes} classification problem involves predicting major events based on averaged hourly readings from the Northern California Earthquake Data Center. Here, we have two classes: `positive' and `negative'. A positive case is defined as a major earthquake (a reading over 5 on the Richter scale) that is not preceded by another major earthquake for at least 512 hours. A negative case is defined as a reading below 4, preceded by at least 20 non-zero readings in the previous 512 hours. The classification problem is to predict whether a major earthquake is about to occur based on the most recent seismic readings. In this dataset, $k$NN, KDA, CART, Random Forest, nonlinear SVM and the LMD classifier have the same misclassification rate of 25.18\%. GLMNET has an error rate of 26.62\%. Misclassification rates of other classifiers are slightly higher.

The \textbf{Lightning~2} dataset consists of 60 training and 61 test observations on two classes: `Cloud-to-Ground lightning' and `Intra-Cloud lightning'. The 637-dimensional observations were obtained from recordings of transient electromagnetic events by FORTE satellite with a two-step process. First, a Fourier transformation was used to get the spectrogram, which were collapsed to give a power density time series. These power density time series were finally smoothed to get the observations. In this example, $k$NN and KDA have the lowest misclassification rates (24.6\%). GLMNET, Random Forest and linear SVM have the second lowest misclassification rate of 26.23\%. MD and LMD classifiers have the next best performance with a misclassification rate of 27.87\%.

\begin{table}[!t]
\setlength{\tabcolsep}{0.03in}
    \centering
    \small
    \begin{tabular}{ccccccccccccc}
       Data set & LDA & QDA & GLM  & $k$NN & KDA & CART & Random & SVM & SVM & MD & LMD \\ 
        & &    &  NET  &  &   &  & Forest & Linear & RBF & \multicolumn{2}{c}{(proposed)}    \\ \hline

          Synthetic & 10.80 & 10.20  & 11.00 & \best{8.40} & \best{8.40} & 11.90 & 10.40 & 10.70 & 9.20 & 10.20 & 10.20\\ 
           & \se{0.98} & \se{0.96} & \se{0.99} & \se{0.88} & \se{0.88} & \se{1.02} & \se{0.97} & \se{0.98} & \se{0.91} & \se{0.96} & \se{0.96} \\ \hline

          Blood  & 22.66 & 22.43 & 22.65 &23.27 &23.04 & 22.57& 22.67 & 23.45 &22.09 & \best{21.91} &22.15\\ 
          Transfusion & \se{0.07} & \se{0.13}  & \se{0.07} & \se{0.14} & \se{0.16} & \se{0.15} & \se{0.15} & \se{0.03} & \se{0.16} & \se{0.11} & \se{0.09}\\ \hline

          Iris& \best{2.55} & 2.81  &5.03 & 4.72 & 4.39 & 7.32 &5.27 & 3.41 &5.68 & 3.99 & 4.31 \\
        & \se{0.16} & \se{0.17} & \se{0.21} & \se{0.26} & \se{0.18} & \se{0.57} & \se{0.18} & \se{0.18} & \se{0.24} & \se{0.23} & \se{0.26}\\ \hline

          Biomed & 15.01 &12.05 & 11.21 & 13.91 &14.91 &19.40 &13.10 & {\bf 11.05} &13.12 &12.49 & 12.72 \\
            & \se{0.31} & \se{0.27} & \se{0.27} & \se{0.30} & \se{0.40} & \se{0.44} & \se{0.30} & \se{0.24} & \se{0.30} & \se{0.27} & \se{0.27} \\ \hline
            
          Pima  Indian& 23.38 & 25.93 & 23.58 & 27.32 & 28.97 & 26.03 & 24.17 & \best{23.30} & 24.46 & 24.93 & 24.85\\
          Diabetes & \se{0.17} & \se{0.19} & \se{0.18} & \se{0.19} & \se{0.20} & \se{0.20} & \se{0.15} & \se{0.17} & \se{0.16} & \se{0.17} & \se{0.18} \\ \hline
       
          Vehicle &   22.57 & \best{16.47} & 21.02 & 37.53 & 36.47 & 32.58 & 25.96 & 21.46 &19.63 & 17.72 & 18.73\\
        & \se{0.14} & \se{0.15}  & \se{0.17} & \se{0.20} & \se{0.21} & \se{0.24} & \se{0.16} & \se{0.16} & \se{0.18} & \se{0.22} & \se{0.31} \\ \hline
        
        Italy Power & 8.55 & 10.30 &  3.11 & 3.69 & 4.57 & \best{3.01} & 4.57 & 3.89 & 6.12 & 6.03 & 5.25\\ 
        Demand & \se{0.87} & \se{0.95} & \se{0.54} & \se{0.59} & \se{0.65} & \se{0.53} & \se{0.65} & \se{0.6} & \se{0.75} & \se{0.74} & \se{0.70}\\ \hline

         Statlog  & 17.15 & 15.20 & 15.90& 9.60 & 9.70 &14.80 & 9.35 & 14.05 & \best{9.15} & 12.50 & 13.50 \\  
         (Landsat Satellite) &  \se{0.84} & \se{0.80} &  \se{0.82} & \se{0.66} & \se{0.66} & \se{0.79} & \se{0.65} & \se{0.78} & \se{0.64} & \se{0.74} & \se{0.76} \\ \hline  
         
        Control Chart & 0.59 & \best{0.00}  & 0.30 & 0.34 & 0.65 & 2.23 & \best{0.00} &\best{0.00} & \best{0.00} & \best{0.00} & \best{0.00} \\
        
        & \se{0.08} & \se{0.00} & \se{0.07} & \se{0.05} & \se{0.10} & \se{0.17} & \se{0.00} & \se{0.00} &\se{0.00} & \se{0.00} & \se{0.00}\\ \hline

        Sony AIBO & \best{6.49} & \best{6.49}  & 33.44 & 30.45 & 30.45 & 57.07 & 37.27 & 32.28 & 19.63 & 16.81 & \best{6.49} \\
       Robot Surface 1  & \se{1.00} & \se{1.00} & \se{1.92} & \se{1.88} & \se{1.88} & \se{2.02} & \se{1.97} & \se{1.91} & \se{1.62} & \se{1.53} & \se{1.00}\\ \hline

         Trace & 34.00 & 34.00 & 27.00 & 24.00 & 27.00 & 37.00 & 20.00 & 19.00 & 16.00 & \best{8.00} & \best{8.00} \\ 
         &\se{4.74} & \se{4.74} & \se{4.44} & \se{4.27} & \se{4.44} & \se{4.83} & \se{4} & \se{3.92} & \se{3.67} & \se{2.71} & \se{2.71}\\ \hline

         Ham & 23.81 & 25.71  & 34.29 & 40.00 & 40.00 & 29.52 & 27.62 & 38.10 & 26.67 & \best{21.90} & 22.86 \\
         & \se{4.16} & \se{4.27} & \se{4.63} & \se{4.78} & \se{4.78} & \se{4.45} & \se{4.36} & \se{4.74} & \se{4.32} & \se{4.04} & \se{4.10}\\ \hline

         Meat & 6.67 &5.00 & 5.00 & 8.33 & 8.00 & 16.67 & 8.00 & 6.67 & 10.00 & \best{3.33} & \best{3.33} \\ 
         & \se{3.22} & \se{2.81} & \se{2.81} & \se{3.57} & \se{3.50} & \se{4.81} & \se{3.50} & \se{3.22} & \se{3.87} & \se{2.32} & \se{2.32}\\ \hline

        Earthquakes & 29.50 & 32.37 & 26.62 & \best{25.18} & \best{25.18} & \best{25.18} & \best{25.18} & 35.97 & \best{25.18} & 30.94 & \best{25.18}\\
          & \se{3.87} & \se{3.97} & \se{3.75} & \se{3.68} & \se{3.68} & \se{3.68} & \se{3.68} & \se{4.07} & \se{3.68} & \se{3.92} & \se{3.68}\\ \hline
         
          Lightning 2&  32.79 & 32.79  & 26.23 & \best{24.60}  & \best{24.60} &37.71
          & 26.23 & 26.23 & 29.51 & 27.87 & 27.87\\ 
          &  \se{6.01} & \se{6.01}  & \se{5.63} & \se{5.51} & \se{5.51} & \se{6.21} & \se{5.63} & \se{5.63} & \se{5.84} & \se{5.74} & \se{5.74}\\ \hline 
          
          Worms Two-class & 48.05 & 48.05  & 48.05 & 38.96 & 38.96 & 41.56 & \best{35.06} & 46.75 & 36.36 & 46.75 & 36.36 \\
         & \se{5.69} & \se{5.69} & \se{5.69} & \se{5.56} & \se{5.56} & \se{5.62} & \se{5.44} & \se{5.69} & \se{5.48} & \se{5.69} & \se{5.48}\\ \hline
 
        Colon Cancer & \best{13.69} & 21.78 & 23.53 & 24.28 & 27.28 & 33.72 & 27.09 & 17.12 & 19.03 & 21.75 & 22.44 \\
        
        & \se{0.48} & \se{0.68}  & \se{1.20} & \se{0.61} & \se{0.46} & \se{0.65} & \se{0.62} & \se{0.56} & \se{0.82} & \se{0.76} & \se{0.76} \\ \hline

          House Twenty & 23.53 & 42.02 & 32.77 & 29.41 & 42.02 & 45.38 & 22.69 & 27.73 & 26.89 & \best{21.85} & 24.37 \\
         & \se{3.89} & \se{4.52} & \se{4.30} & \se{4.18} & \se{4.52} & \se{4.56} & \se{3.84} & \se{4.10} & \se{4.06} & \se{3.79} & \se{3.94}\\ \hline
    \end{tabular}

    Boldface character signifies the best result in each case.
    \caption{Average misclassification rates (in \%) of different classifiers in the benchmark datasets. The standard errors are reported in the next line within brackets in a smaller font.\label{tab:benchmark_results}}
    \vspace{-0.1in}
\end{table}

The \textbf{Worms Two-class} dataset consists of 258 traces of worms that have been transformed into four eigenworm series. This dataset focuses on the behavioral genetics of Caenorhabditis elegans, a roundworm commonly used in genetics research. The motion of the worms is represented by four scalar values, indicating the amplitudes along each dimension when the shape is projected onto eigenworms. There are two classes: N2 in one class (Wild), and mutant types including goa-1, unc-1, unc-38, and unc-63 in the second class (Mutant). The goal is to classify individual worms as Wild or Mutant based on the time series of the first eigenworm. The Random Forest classifier has the lowest misclassification rate (35.06\%) in this dataset closely followed by nonlinear SVM and the LMD classifier (36.36\%). Among the others, only $k$NN and KDA have misclassification rates lower than 40\%.

The \textbf{Colon Cancer} dataset contains information related to micro-array expression levels of 2000 genes from colon tissues, where we aim to classify the tissues into two classes: `normal' and `colon cancer' based on the expression levels of genes \citep{alon1999broad}. This dataset has a good linear separation among the observations from the two competing classes. As a result, linear classifiers like LDA and linear SVM have lower misclassification rates. Among the other classifiers, QDA, GLMNET, $k$NN, nonlinear SVM, MD, and LMD classifiers have satisfactory performance.

The \textbf{House Twenty} data collection was a part of the REFIT project, which aimed to develop personalized retrofit decision support tools for UK homes using smart home technology. It was created by \cite{murray2015data}. The dataset includes information from 20 households in the Loughborough area from 2013 to 2014. There are two classes: `household usage of electricity' and `electricity of tumble dryer and washing machine' \citep[see][for details]{murray2015data}. Here, the MD classifier has the minimum misclassification rate followed by Random Forest and LDA. Among the rest, both linear and nonlinear SVM and the LMD classifier perform better.

\begin{figure}[t]
\begin{center} 
\includegraphics[width=0.9\linewidth,,height=3.0in]{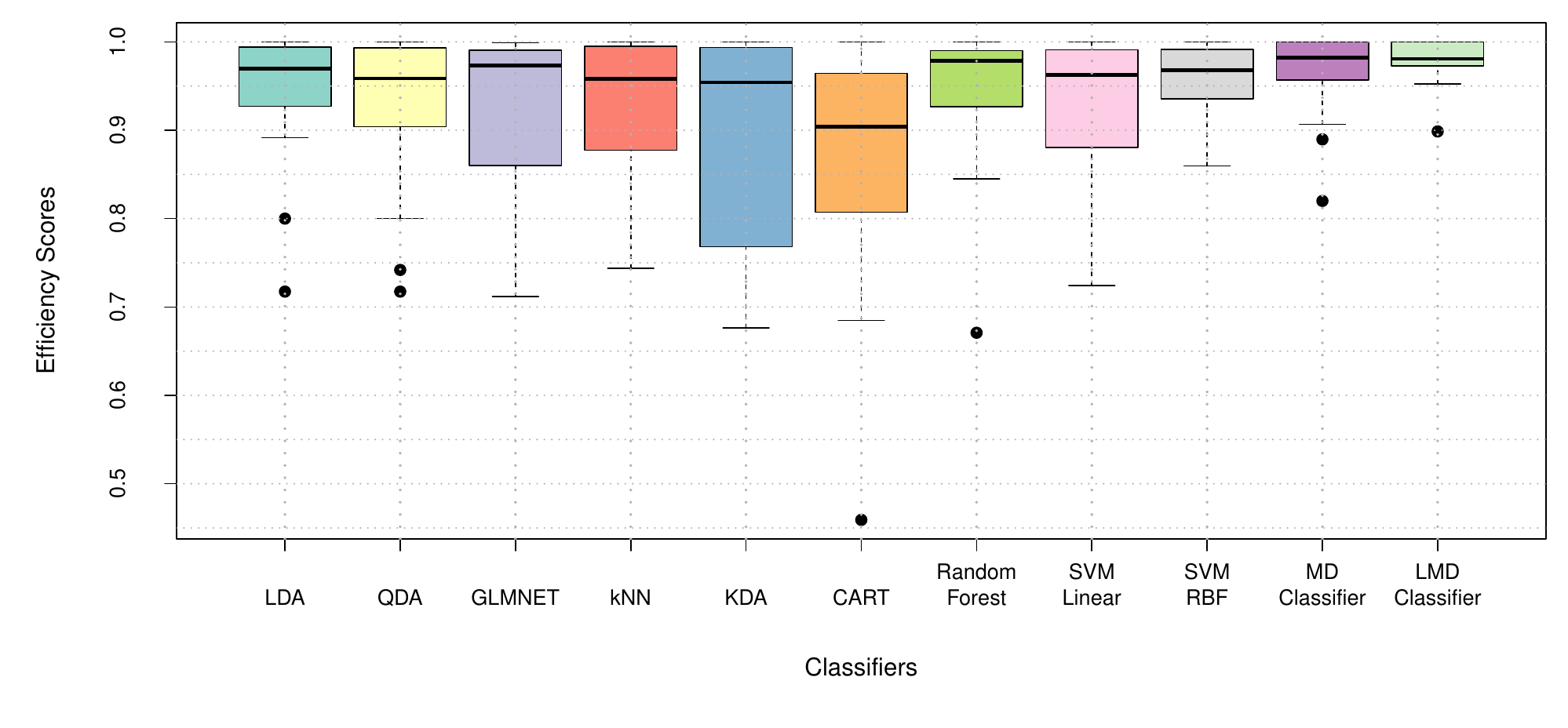}
\vspace{-0.25 in}
\caption{Boxplot of the efficincy scores of different classifiers on benchmark datasets.\label{fig:MD_LMD_boxplot}}
\end{center}
\vspace{-0.35in}
\end{figure}

To compare the overall performance of different classifiers on these benchmark datasets, following the idea \cite{ghosh2012hybrid}, we use a notion of efficiency. For a given dataset, if the correct classification rate of $T$ classifiers are $C_1,\ldots,C_T$, then the efficiency score of the $t$-th classifier is defined as $e_t=C_t/\max_{1 \le t \le T} C_t$. In any example, the best classifier has $e_t = 1$, while a value close to $1$ shows that the classifier performs similarly to the best classifier. On the other hand, a low value of $e_t$ indicates a poor relative performance (i.e., lack of efficiency) by the corresponding classifier. In each of these benchmark examples, we compute this ratio for all classifiers, and they are graphically represented by box plots in Figure~\ref{fig:MD_LMD_boxplot}. This figure clearly shows that the LMD classifier has the best overall performance. The MD classifier also outperforms most of the state-of-the-art classifiers considered in this article. Only Random Forest and nonlinear SVM have somewhat competitive overall performances.

\section{Concluding Remarks\label{sec:conclusion}}

In this article, we have proposed some classifiers based on Mahalanobis distances and local Mahalanobis distances. While popular parametric classifiers like LDA and QDA are mainly motivated by the normality of the underlying distributions, the MD classifier works well for a much broader class. Unlike logistic regression, GLMNET and linear SVM, it does not assume any linear form of the discriminating surface. So, if the Bayes class boundary is highly nonlinear, the MD classifier often performs much better than the linear classifiers. If the underlying distributions are elliptic, it usually outperforms popular nonparametric classifiers, especially when the sample size is small compared to the dimension of the data. However, in cases of non-elliptic, and more specifically multimodal distributions, the MD classifier may fail to yield satisfactory performance. The LMD classifier takes care of this problem. The data-driven choice of the localization parameter $h$ makes this classifier more flexible. For large values of $h$, it behaves like the MD classifier, which works well when the competing classes are nearly elliptic. At the same time, the use of small values of $h$ helps to cope with non-elliptic and multimodal distributions. For choosing the value of $h$, we have used the bootstrap method, which worked well in our numerical studies. However, one can use cross-validation or other resampling techniques \citep[e.g.,][]{ghosh2008error} as well. Moreover, for the LMD classifier, we have used the same tuning parameter for all the classes. It is possible to use different tuning parameters for different classes, albeit at the expense of an additional computation, which grows exponentially with the number of classes. In our numerical studies, the use of different tuning parameters (chosen by bootstrap) did not make any significant difference in the misclassification rates. Instead of choosing a particular value of $h$, one can also adopt a multiscale approach \citep[e.g.,][]{holmes2002probabilistic, holmes2003likelihood, ghosh2005visualization, ghosh2006classification} and judiciously aggregate the results obtained for several choices $h$ to come up with the final decision. Following a similar idea, we can use LMD computed for different values of $h$ as features and fit a generalized additive model based on them. The number of features can be large, and we can use a penalized method with LASSO or elastic net penalty \citep[e.g.,][]{hastie2015statistical} to estimate the posterior probabilities of different classes using a parsimonious model based on a lesser number of selected features.

Due to data sparsity in high dimensions, when many nonparametric classifiers perform poorly, the proposed MD and LMD classifiers can have excellent performance. Analyzing several simulated datasets, we have amply demonstrated these features of the proposed classifiers in this article. Analyses of benchmark datasets also show that our proposed classifiers can be at par or better than the state-of-the-art classifiers in a wide variety of classification problems.

\section*{Acknowledgement}
 The authors would like to thank the Editor-in-chief, Associate Editor, and an anonymous reviewer for their insightful comments, which led to an improved version of the manuscript. The research of Soham Sarkar was partially supported by the INSPIRE Faculty Fellowship from the Department of Science and Technology, Government of India. 

\appendix
\section*{Appendix: Proofs and Mathematical Details}

\begin{proof}[Proof of Theorem~\ref{thm:MD_GAM}]
Recall that the density $f_j$ of an elliptic distribution with location $\muvec_j$ and scatter matrix $\sigmat_j$ can be expressed as
\[
f_j({\bf x}) = C_j |\sigmat_j|^{-1/2} \phi_j(\delta_j({\bf x})),
\]
where $\delta_j({\bf x}) = \{({\bf x} -\muvec_j )^\top\sigmat_j^{-1}({\bf x} -\muvec_j )\}^{1/2}$.  Let $\pi_{j}$ be the prior probability of the $j$-th class. Then, for $j=1,\ldots,J-1$, we have
\[
\log\left(\frac{p(j \mid \vec x)}{p(J \mid \vec x)}\right) = \log\left(\frac{\pi_{j}f_{j}(\mathbf{x})}{\pi_{J}f_{J}(\mathbf{x})}\right) =(\alpha_j - \alpha_J) + \psi_j(\delta_j ({\bf x})) - \psi_J (\delta_J({\bf x})),
\]
where $\psi_j(t)= \log(\phi_j(t))$ and $\alpha_j= \log(C_j \pi_j |\sigmat_j|^{-1/2})$. Now, define $g_j(\Delta({\bf x)})=\sum_{k} g_{jk}(\delta_j(\vec x))$, where 
\[
g_{jk}(t)=\begin{cases}\alpha_j + \psi_j(t) & \text{if } k=j,\\
-(\alpha_J + \psi_J(t)) & \text{if } k=J,\\
0 & \text{if } k \neq {j, J}.
\end{cases}
\]
Hence, $\log\big(p(j \mid \mathbf{x})/p(J \mid {\bf x})\big) = g_j\big(\Delta(\vec x)\big)$ or $p(j \mid \vec x)=\exp\big\{g_j\big(\Delta(\vec x)\big)\big\}\,p(J \mid \vec x)$ for $j=1,2,\ldots,J-1$,
where $g_j(\cdot)$ is an additive function. The result now follows by noting that the sum of all posterior probabilities is $1$.
\end{proof}

\begin{proof}[Proof of Lemma~\ref{lemma:MD_uniform_convergence}]
Note that
\begin{align*}
&\left|(\vec x-\widehat\muvec)^\top \widehat\sigmat^{-1}(\vec x-\widehat\muvec) - (\vec x-\muvec)^\top \sigmat^{-1}(\vec x-\muvec)\right| \\
&=\left|(\vec x-\widehat\muvec)^\top \widehat\sigmat^{-1}(\vec x-\widehat\muvec) - (\vec x-\widehat\muvec)^\top \sigmat^{-1}(\vec x-\widehat\muvec) + (\vec x-\widehat\muvec)^\top \sigmat^{-1}(\vec x-\widehat\muvec) - (\vec x-\muvec)^\top \sigmat^{-1}(\vec x-\muvec) \right| \\
&=\left|(\vec x-\widehat\muvec)^\top (\widehat\sigmat^{-1}\mkern-2mu - \sigmat^{-1}) (\vec x-\widehat\muvec) + (\muvec - \widehat\muvec)^\top \sigmat^{-1}(2\vec x-\muvec-\widehat\muvec) \right| \\
&\le \left|(\vec x-\widehat\muvec)^\top(\widehat\sigmat^{-1}\mkern-2mu-\sigmat^{-1})(\vec x-\widehat\muvec)\right| + \left|(\widehat\muvec-\muvec)^\top \sigmat^{-1} (2\vec x-\widehat\muvec-\muvec)\right|.
\end{align*}
The second term on the right can be bounded by $\lambda_{\max}(\sigmat^{-1})\,\|\widehat\muvec-\muvec\|\,\big(2\|\vec x\|+\|\widehat\muvec\|+\|\muvec\|\big)$, where $\lambda_{\max}$ denotes the largest eigenvalue of a matrix. Since $\|\widehat\muvec - \muvec\| = O_P(n^{-1/2})$, it follows that $\|\widehat\muvec\| = O_P(1)$. Consequently, for any compact set $\mathcal C \subset \R^d$, $\sup_{{\bf x} \in {\cal C}} \big|(\widehat\muvec - \muvec)^\top \sigmat^{-1} (2 \vec x - \widehat\muvec - \muvec)\big| = O_P(n^{-1/2})$. Now, for the first term on the right in the equation above, we can find an upper bound $\max_{i=1,\ldots,d} |\lambda_i(\widehat\sigmat^{-1}\mkern-2mu-\sigmat^{-1})|\,\|\vec x - \widehat\muvec\|^2$, where $\lambda_i$ denotes the $i$-th ($i=1,\ldots,d$) largest eigenvalue of a matrix. Since $\max_{i=1,\ldots,d} |\lambda_i(\widehat\sigmat^{-1}\mkern-2mu-\sigmat^{-1})| \le \|\widehat\sigmat^{-1}\mkern-2mu-\sigmat^{-1}\|_{\rm F} = O_P(n^{-1/2})$, it follows that $\sup_{\vec x \in \mathcal C} |(\vec x-\widehat\muvec)^\top(\widehat\sigmat^{-1}\mkern-2mu-\sigmat^{-1})(\vec x-\widehat\muvec)| = O_P(n^{-1/2})$. The proof of the lemma follows by combining the two parts.
\end{proof}

\begin{proof}[Proof of Lemma~\ref{lemma:LMD_asymptotics}]
(a) Since $K$ is continuous at $\vec 0$, for any $\vec x \in \R^d$,
\[
K\Big(\frac{\sigmat^{-1/2}(\vec x-\vec X)}{h}\Big) \to K(\vec 0) \qquad \text{as } h \to \infty.
\]
Moreover, since $K$ is spherically symmetric about $\vec 0$ and $K(\vec t) \le K(\vec 0)$ for all $\vec t \in \R^d$, part (a) of the Lemma now follows from a simple application of the Dominated Convergence Theorem.

\medskip
\noindent
(b) For small $h$, $\beta_h(\vec x)$ can be expressed as
\begin{align*}
\beta_{h}(\vec x) &= \int K\Big(\frac{\sigmat^{-1/2}(\vec x-\vec y)}{h}\Big)(\vec x-\vec y)^\top\sigmat^{-1}(\vec x-\vec y)\,f(\vec y)\,d\vec y\\
&= h^{d+2}\,|\sigmat|^{1/2} \int \vec z^\top\vec z\,K(\vec z)\,f(\vec x-h\sigmat^{1/2}\vec z)\,d\vec z \kern10ex \Big(\text{substituting } \vec z=\frac{\sigmat^{-1/2}(\vec x-\vec y)}{h}\Big)\\
 &= h^{d+2}\,|\sigmat|^{1/2} \int \vec z^\top\vec z\,K(\vec z)\big\{f(\vec x)-h{\vec z}^{\top}\sigmat^{1/2}\nabla f(\xivec)\big\}\,d\vec z,\mkern10mu \text{(using first order Taylor expansion)}\\
&\kern45ex \mbox{ where } \xivec=\alpha \vec x +(1-\alpha)\vec z \text{ for some } \alpha \in (0,1).
\end{align*}

\noindent
Now, using Cauchy-Schwarz inequality, one gets 
\[
\Big|\int \vec z^\top\vec z\,K(\vec z)\,{\vec z}^{\top}\sigmat^{1/2}\nabla f(\xivec)\,d\vec z\Big| 
\le \sqrt{\lambda_{\text{max}}(\sigmat)} \, \sup_{\vec x} \|\nabla f(\vec x)\| 
\int \|\vec z\|^3\,K(\vec z)\,d\vec z .
\]
So, if $\sup_x \|\nabla f(\vec x)\|<\infty$ and $\int \|\vec z\|^3 K(\vec z)\,d\vec z <\infty$,  as $h \to 0$, we have
\[
\beta_{h}(\vec x)/h^{d+2} = |\sigmat|^{1/2} f(\vec x) \int \vec z^\top\vec z\,K(\vec z)\,d\vec z +O(h) \rightarrow |\sigmat|^{1/2} f(\vec x) \kappa_2 \quad\text{as } h \to 0. \qedhere
\]
\end{proof}

\begin{proof}[Proof of Theorem~\ref{thm:LMD_GAM}]
From Lemma~\ref{lemma:LMD_asymptotics}, we have $\gamma_j^0(\vec x) = \lim_{h \downarrow 0} \gamma_j^{h}(\vec x) = |\sigmat_j|^{1/2}\kappa_2\,f_j({\bf x})$ for $j=1,\ldots,J$. So, $f_j(x) = \kappa_2^{-1} |\sigmat_j|^{-1/2} \gamma_j^0(\vec x)$. Now, for $j=1,\ldots,J-1$
\[
\log\left(\frac{p(j \mid \vec x)}{p(J \mid \vec x)}\right) = \log\left(\frac{\pi_{j}f_{j}(\vec x)}{\pi_{J}f_{J}(\vec x)}\right) = (b_j - b_J) + \log\big(\gamma^0_j(\vec x)\big) - \log\big(\gamma^0_J(\vec x)\big),
\]
where $b_j= \log(\pi_j |\sigmat_j|^{-1/2})$. The rest of the proof is similar to the proof of Theorem~\ref{thm:MD_GAM}.
\end{proof}

\begin{proof}[Proof of Lemma~\ref{lemma:LMD_uniform_convergence}]
Recall that for any $h>0$, 
\begin{align*}
    \beta_h(\vec x) &= E\left\{K\Big(\frac{\sigmat^{-1/2}(\vec x-\vec X)}{h}\Big)(\vec x-\vec X)^\top\sigmat^{-1}(\vec x-\vec X)\right\}\text{ and } \\
    \widehat\beta_h(\vec x) &= \frac{1}{n}\sum_{i=1}^{n} K\Big(\frac{\widehat\sigmat^{-1/2}(\vec x-\vec X_i)}{h}\Big)(\vec x-\vec X_i)^\top\widehat\sigmat^{-1}(\vec x-\vec X_i).
\end{align*}
Therefore, $|{\widehat\beta}_h({\bf x})-\beta_h({\bf x})|$ can be expressed as
\begin{align*}
& \left|\frac{1}{n}\sum_{i=1}^{n} K\Big(\frac{\widehat\sigmat^{-1/2}(\vec x-\vec X_i)}{h}\Big)(\vec x-\vec X_i)^\top\widehat\sigmat^{-1}(\vec x-\vec X_i) - E\left\{K\Big(\frac{\sigmat^{-1/2}(\vec x-\vec X)}{h}\Big)(\vec x-\vec X)^\top\sigmat^{-1}(\vec x-\vec X)\right\}\right|\\
&= \Bigg|\frac{1}{n}\sum_{i=1}^{n} \Psi\bigg(\frac{(\vec x-\vec X_i)^\top\widehat\sigmat^{-1}(\vec x-\vec X_i)}{h^2}\bigg)(\vec x-\vec X_i)^\top\widehat\sigmat^{-1}(\vec x-\vec X_i)\\
&\kern40ex -E\left\{\Psi\bigg(\frac{(\vec x-\vec X)^\top\sigmat^{-1}(\vec x-\vec X)}{h^2}\bigg)(\vec x-\vec X)^\top\sigmat^{-1}(\vec x-\vec X)\right\}\Bigg|\\
& \le \Bigg|\frac{1}{n}\sum_{i=1}^{n} \Psi\bigg(\frac{(\vec x-\vec X_i)^\top\widehat\sigmat^{-1}(\vec x-\vec X_i)}{h^2}\bigg)(\vec x-\vec X_i)^\top\big(\widehat\sigmat^{-1}\mkern-2mu-\sigmat^{-1})(\vec x-\vec X_i)\Bigg|\\
&\kern3ex + \Bigg|\frac{1}{n}\sum_{i=1}^{n} \bigg\{\Psi\bigg(\frac{(\vec x-\vec X_i)^\top\widehat\sigmat^{-1}(\vec x-\vec X_i)}{h^2}\bigg)-\Psi\bigg(\frac{(\vec x-\vec X_i)^\top\sigmat^{-1}(\vec x-\vec X_i)}{h^2}\bigg)\bigg\}(\vec x-\vec X_i)^\top\sigmat^{-1}(\vec x-\vec X_i) \Bigg|\\
&\kern3ex+ \Bigg|\frac{1}{n}\sum_{i=1}^{n}\Psi\bigg(\frac{(\vec x-\vec X_i)^\top\sigmat^{-1}(\vec x-\vec X_i)}{h^2}\bigg)(\vec x-\vec X_i)^\top\sigmat^{-1}(\vec x-\vec X_i)\\
&\kern40ex-E\left\{\Psi\bigg(\frac{(\vec x-\vec X)^\top\sigmat^{-1}(\vec x-\vec X)}{h^2}\bigg)(\vec x-\vec X)^\top\sigmat^{-1}(\vec x-\vec X)\right\}\Bigg|\\
&=: A_{1,n}(\vec x) + A_{2,n}(\vec x) + A_{3,n}(\vec x).
\end{align*}
If $\Psi$ is bounded, i.e., $\sup_t |\Psi(t)| \le M$ for some $M>0$, following the proof of Lemma~\ref{lemma:MD_uniform_convergence}, we get
\[
A_{1,n}(\vec x) \le M \max_{i=1,\ldots,d}\big|\lambda_i(\widehat\sigmat^{-1}\mkern-2mu-\sigmat^{-1})\big|\,\frac{1}{n}\sum_{i=1}^{n} \|\vec x-\vec X_i\|^2=O_P(n^{-1/2}) \Big\{\|\vec x\|^2 + \frac{1}{n}\sum_{i=1}^{n} \|\vec X_i\|^2\Big\}.
\]
Since $n^{-1} \sum_{i=1}^n \|\vec X_i\|^2 = O_P(1)$, it follows that for any compact set $\mathcal C$, $\sup_{\vec x \in \mathcal C} A_{1,n}(\vec x) =O_p(n^{-1/2})$. Next, since $\Psi$ is Lipschitz continuous, we have
\begin{align*}
A_{2,n}(\vec x) &\le \frac{1}{n} \sum_{i=1}^{n} L \frac{1}{h^2}\Big|(\vec x-\vec X_i)^\top(\widehat\sigmat^{-1}\mkern-2mu-\sigmat^{-1})(\vec x-\vec X_i)\Big|\,(\vec x-\vec X_i)^\top\sigmat^{-1}(\vec x-\vec X_i)\\
&\le \frac{L}{h^2} \max_{i=1,\ldots,d}\big|\lambda_i(\widehat\sigmat^{-1}\mkern-2mu-\sigmat^{-1})\big|\,\lambda_{\max}(\sigmat^{-1})\,\frac{1}{n}\sum_{i=1}^{n} \|\vec x-\vec X_i\|^4 = O_P(n^{-1/2}) \Big\{\|\vec x\|^4 + \frac{1}{n}\sum_{i=1}^{n}\|\vec X_i\|^4\Big\},
\end{align*}
where $L$ is the Lipschitz constant for $\Psi$. Since $\vec X$ has finite fourth moments, we have $n^{-1}\sum_{i=1}^n \|\vec X_i\|^4 = O_P(1)$, from which it follows that $\sup_{\vec x \in \mathcal C} A_{2,n}(\vec x) = O_P(n^{-1/2})$. Finally, $A_{3,n}$ is of the form $n^{-1}\sum_{i=1}^n Y_i - \E(Y_1)$, where $Y_1,\ldots,Y_n$ are i.i.d.\ random variables. Moreover, since $\Psi$ is bounded and $\vec X$ has finite fourth moments, it follows that
\begin{align*}
E(Y^2) &= E\left[\bigg\{\Psi\bigg(\frac{(\vec x-\vec X)^\top\sigmat^{-1}(\vec x-\vec X)}{h^2}\bigg)(\vec x-\vec X)^\top\sigmat^{-1}(\vec x-\vec X)\bigg\}^2\right] \\
& \le M^2 E\big[\{(\vec x-\vec X)^\top\sigmat^{-1}(\vec x-\vec X)\}^2] \\
& \le M^2 \lambda_{\max}^2(\sigmat^{-1}) E(\|\vec x-\vec X\|^4) \le 8 M^2\lambda_{\max}^2(\sigmat^{-1})\,\{\|\vec x\|^4 + E(\|\vec X\|^4)\}.
\end{align*}
Therefore, $\sup_{\vec x \in \mathcal C} Var(Y)$ is finite. Hence, using the Central Limit Theorem, we have 
\[
A_{3,n}(\vec x) = O_P\big(n^{-1/2}\sqrt{Var(Y)}\big) = O_P(n^{-1/2}),
\]
uniformly over $\vec x \in \mathcal C$, i.e., $\sup_{\vec x \in \mathcal C} A_{3,n}(\vec x) = O_P(n^{-1/2})$. Combining the three parts, we get 
\[
\sup_{\vec x \in \mathcal C} \big|\widehat\beta_h(\vec x)-\beta_h(\vec x)\big| \le \sup_{\vec x \in \mathcal C} \{A_{1,n}(\vec x) + A_{2,n}(\vec x) + A_{3,n}(\vec x)\} = O_P(n^{-1/2}).
\]
Recall that 
\[
|\widehat\gamma^{h}(\vec x)-\gamma^{h}(\vec x)| = \begin{cases}
    |\widehat\beta^{h}(\vec x)-\beta^{h}(\vec x)| & \text{if } h > 1,\\
    |\widehat\beta^{h}(\vec x)-\beta^{h}(\vec x)|/h^{d+2} & \text{if } h \le 1. \end{cases}
\]
Thus, $\sup_{\vec x \in \mathcal C} |\widehat\gamma^{h}(\vec x)-\gamma^{h}(\vec x)| = O_P(n^{-1/2})$, as intended.
\end{proof}

\begin{proof}[Proof of Theorem~\ref{thm:MD_asymptotics_HDLSS}]
From (A1), for independent random vectors $\vec X \sim F_j$ and $\vec Y, \vec Z \sim F_{j^\prime}$, we have
\begin{align*}
\left|\frac{1}{d}\|{\bf X}-{\bf Y}\|^2 - \frac{1}{d}E(\|{\bf X}-{\bf Y}\|^2)\right| \stackrel{P}{\rightarrow} 0 \mbox{ and }  \left|\frac{1}{d} \langle{\bf X}-{\bf Y},{\bf X}-{\bf Z}\rangle - \frac{1}{d}E(\langle{\bf X}-{\bf Y},{\bf X}-{\bf Z}\rangle)\right| \stackrel{P}{\rightarrow} 0 \mbox{ as } d \rightarrow \infty.
\end{align*}
Now, using (A2), one can verify that $d^{-1} E(\|{\bf X}-{\bf Y}\|^2) \to \nu_{jj^\prime}^2+\sigma_j^2+\sigma_{j^\prime}^2$ as $d \to \infty$. So,
\[
\frac{\|\vec X-\vec Y\|^2}{d} \overset{P}{\to} \begin{cases} \nu_{jj^\prime}^2 + \sigma_j^2 + \sigma_{j^\prime}^2 & \text{if } j \neq j^\prime, \\
                                                2 \sigma_j^2 & \text{if } j = j^\prime. \end{cases}
\]
Again, using (A2), one can show that ${d}^{-1}E(\langle{\bf X}-{\bf Y},{\bf X}-{\bf Z}\rangle)$ converges to $\nu_{jj'}^2+\sigma_j^2$ as $d \rightarrow \infty$. So,
\[
\frac{\langle \vec X-\vec Y,  \vec X-\vec Z \rangle}{d} \overset{P}{\to} \begin{cases} \nu_{jj^\prime}^2 + \sigma_j^2 & \text{if } j \neq j^\prime, \\
                                                \sigma_j^2 & \text{if } j = j^\prime. \end{cases}
\]

\medskip
\noindent
(a) Since we have used $\widehat\muvec_j = \overline{\vec X}_j = n_j^{-1} \sum_{i=1}^{n_j} \vec X_{ji}$ and $\widehat\sigmat_j = \vec I$, the estimated Mahalanobis distance is given by $\widehat\delta_j(\vec X) = \|\vec X - \overline{\vec X}_j\|^2$. So, for any $\vec X_{ji}$ from the $j$-th class,
\begin{align}\label{eq:sample_MD_convergence_proof_1}
d^{-1}\widehat\delta_j(\vec X_{ji}) = d^{-1}\|{\vec X}_{ji} -{\overline{\vec X}}_j\|^2 &= \frac{1}{n_j^2} \left[ \sum_{k \neq i} d^{-1}\|{\vec X}_{ji}-{\vec X}_{jk}\|^2 + \sum_{k \neq k^{\prime} \neq i} d^{-1} \langle {\vec X}_{ji}-{\vec X}_{jk}, {\vec X}_{ji} -{\vec X}_{jk^{\prime}} \rangle \right] \nonumber\\
& \stackrel{P}{\rightarrow} \frac{1}{n_j^2}\Big[(n_j-1)2\sigma_j^2 + (n_j-1)(n_j-2)\sigma_j^2\Big] = \Big(1-\frac{1}{n_j}\Big)\sigma_j^2 = \theta_{jj}.
\end{align}
Similarly, for $j \neq j^{\prime}$, we have 
\begin{align}\label{eq:sample_MD_convergence_proof_2}
d^{-1}\widehat\delta_{j^{\prime}}(\vec X_{ji}) = d^{-1}\|{\vec X}_{ji} -{\overline{\vec X}}_{j^{\prime}}\|^2 &= \frac{1}{n_{j^{\prime}}^2} \left[ \sum_{k} d^{-1}\|{\vec X}_{ji}-{\vec X}_{j^{\prime}k}\|^2 + \sum_{k \neq k^{\prime}} d^{-1} \langle {\vec X}_{ji}-{\vec X}_{j^{\prime}k}, {\vec X}_{ji} -{\vec X}_{j^{\prime}k^{\prime}} \rangle \right] \nonumber\\
&\stackrel{P}{\rightarrow} \frac{1}{n_{j^{\prime}}^2}\Big[n_{j^{\prime}}(\nu_{jj^{\prime}}^2+\sigma_j^2+\sigma_{j^{\prime}}^2) + n_{j^{\prime}}(n_{j^{\prime}}-1)(\nu_{jj^{\prime}}^2+\sigma_j^2)\Big] \nonumber\\
&= \nu_{jj^{\prime}}^2+\sigma_j^2+\frac{\sigma_{j^{\prime}}^2}{n_{j^{\prime}}} = \theta_{jj^\prime}.
\end{align}
Combining \eqref{eq:sample_MD_convergence_proof_1} and \eqref{eq:sample_MD_convergence_proof_2}, we get that $d^{-1}\widehat\Delta(\vec X_{ji}) = d^{-1}\big(\widehat\delta_1(\vec X_{ji}),\ldots,\widehat\delta_J(\vec X_{ji})\big) \overset{P}{\to} (\theta_{j1},\ldots,\theta_{jJ}) = \Theta_j$.

\medskip
\noindent
(b) Now, for a future observation  ${\vec X}$ from the $j$-th class,
\begin{align*}
d^{-1}\widehat\delta_j(\vec X) &= d^{-1}\|{\vec X} -{\overline{\vec X}}_j\|^2=\frac{1}{n_j^2} 
\left[ \sum_{k} d^{-1}\|{\vec X}-{\vec X}_{jk}\|^2 + \sum_{k \neq k^{\prime}} d^{-1} \langle {\vec X}-{\vec X}_{jk}, {\vec X} -{\vec X}_{jk^{\prime}} \rangle \right] \\
& \stackrel{P}{\rightarrow} \frac{1}{n_j^2}\Big[2n_j\sigma_j^2 + n_j(n_j-1)\sigma_j^2\Big] = \Big(1+\frac{1}{n_j}\Big)\sigma_j^2
\end{align*}
For $j \neq j^{\prime}$, the convergence of $d^{-1} \widehat\delta_{j^{\prime}}(\vec X)$ follows from part (a), which in turn proves the probability convergence of ${\widehat \Delta}({\bf X})$ to $\Theta_j^{\ast}$.
\end{proof}

\begin{proof}[Proof of Theorem~\ref{thm:LMD_asymptotics_HDLSS}]
Recall that $\widehat\Gamma_h(\vec x) = \big(\widehat\gamma^{h}_1(\vec x),\ldots,\widehat\gamma^{h}_J(\vec x)\big)$, where 
\[
\widehat\gamma^{h}_j(\vec x) = \begin{cases} \widehat\beta_{h,j}(\vec x) & \text{if } h > 1, \\
                               \widehat\beta_{h,j}(\vec x)/h^{d+2} & \text{if } h \le 1. \end{cases}
\]
For $\widehat\sigmat_j = \vec I$ and large $h$ (here $h$ is of the same asymptotic order as $d^{1/2}$), $\widehat\gamma_{h,j}=\widehat\beta_{h,j}$ is given by
\[
\widehat\gamma_{h,j}(\vec x) = \frac{1}{n_j} \sum_{i=1}^{n_j} \Psi\left(\frac{\|\vec x - \vec X_{ji}\|^2}{h^2}\right) \|\vec x - \vec X_{ji}\|^2.
\]
From the first part in the proof of Theorem~\ref{thm:MD_asymptotics_HDLSS}, it follows that if $\vec X \sim F_j$ and $\vec Y \sim F_{j^\prime}$ are independent, then
\begin{align*}
\frac{1}{d} \Psi\left(\frac{\|\vec X-\vec Y\|^2}{h^2}\right)\|\vec X-\vec Y\|^2 &= \Psi\left(\frac{\|\vec X-\vec Y\|^2}{d} \frac{d}{h^2}\right)\,\frac{\|\vec X-\vec Y\|^2}{d} \\
& \stackrel{P}{\rightarrow} \Psi\left(\frac{\sigma_j^2+\sigma_{j^\prime}^2+\nu_{jj^\prime}^2}{C_0}\right) (\sigma_j^2+\sigma_{j^\prime}^2+\nu_{jj^\prime}^2) =: \theta_{jj^\prime}^{0}.
\end{align*}
Here, the last step follows from the continuity of $\Psi$. Since the sample sizes $n_1,\ldots,n_J$ are fixed, as $d \to \infty$, we get 
\[
d^{-1} \widehat\gamma_{k,h}(\vec X_{ji}) = \frac{1}{n_k} \sum_{i^\prime=1}^{n_k} \Psi\left(\frac{\|\vec X_{ji}-\vec X_{ki^\prime}\|^2}{h^2}\right) \frac{\|\vec X_{ji}-\vec X_{ki^\prime}\|^2}{d} \overset{P}{\to} \begin{cases} \theta_{kj}^{0} & \text{if } k \ne j,\\ \Big(1-\frac{1}{n_j}\Big)\theta_{jj}^{0} & \text{if } k=j. \end{cases}
\]
For this last step, note that for any $k=1,\ldots,J$, one of the summands in $\widehat\beta_{h,k}(\vec X_{ji})$ is zero if $k=j$. Note that since $h^2/d \to C_0>0$, if follows that $h>1$ for large enough $d$. Part (a) of the theorem now follows easily. Part (b) of the theorem also follows in a similar manner by noting that none of the summands are zero in $\widehat\beta_{h,k}(\vec Z)$ for an independent observation $\vec Z$ outside of the training sample.
\end{proof}

\bibliographystyle{chicago}
\bibliography{biblio}

\begin{thebibliography}{}

\bibitem[\protect\citeauthoryear{Ahn, Marron, Muller, and Chi}{Ahn
  et~al.}{2007}]{ahn2007high}
Ahn, J., J.~S. Marron, K.~M. Muller, and Y.-Y. Chi (2007).
\newblock The high-dimension, low-sample-size geometric representation holds
  under mild conditions.
\newblock {\em Biometrika\/}~{\em 94\/}(3), 760--766.

\bibitem[\protect\citeauthoryear{Alcock and Manolopoulos}{Alcock and
  Manolopoulos}{1999}]{alcock1999time}
Alcock, R.~J. and Y.~Manolopoulos (1999).
\newblock Time-series similarity queries employing a feature-based approach.
\newblock In {\em 7th Hellenic Conference on Informatics}, pp.\  27--29.

\bibitem[\protect\citeauthoryear{Alon, Barkai, Notterman, Gish, Ybarra, Mack,
  and Levine}{Alon et~al.}{1999}]{alon1999broad}
Alon, U., N.~Barkai, D.~A. Notterman, K.~Gish, S.~Ybarra, D.~Mack, and A.~J.
  Levine (1999).
\newblock Broad patterns of gene expression revealed by clustering analysis of
  tumor and normal colon tissues probed by oligonucleotide arrays.
\newblock {\em Proceedings of the National Academy of Sciences\/}~{\em
  96\/}(12), 6745--6750.

\bibitem[\protect\citeauthoryear{Anderson}{Anderson}{2009}]{anderson1962introduction}
Anderson, T.~W. (2009).
\newblock {\em An Introduction to Multivariate Statistical Analysis}.
\newblock Wiley, New York.

\bibitem[\protect\citeauthoryear{Aoshima and Yata}{Aoshima and
  Yata}{2018}]{aoshima2018two}
Aoshima, M. and K.~Yata (2018).
\newblock Two-sample tests for high-dimension, strongly spiked eigenvalue
  models.
\newblock {\em Statistica Sinica\/}~{\em 28\/}(1), 43--62.

\bibitem[\protect\citeauthoryear{Banerjee and Ghosh}{Banerjee and
  Ghosh}{2025}]{banerjee2022high}
Banerjee, B. and A.~K. Ghosh (2025).
\newblock On high dimensional behaviour of some two-sample tests based on ball
  divergence.
\newblock {\em Statistica Sinica\/}~{\em To appear}.
\newblock DOI: 10.5705/ss.202023.0069.

\bibitem[\protect\citeauthoryear{Breiman}{Breiman}{2001}]{breiman2001random}
Breiman, L. (2001).
\newblock Random forests.
\newblock {\em Machine Learning\/}~{\em 45}, 5--32.

\bibitem[\protect\citeauthoryear{Breiman, Friedman, Olsen, and Stone}{Breiman
  et~al.}{2017}]{breiman2017classification}
Breiman, L., J.~H. Friedman, R.~Olsen, and C.~Stone (2017).
\newblock {\em Classification and Regression Trees}.
\newblock Routledge, New York.

\bibitem[\protect\citeauthoryear{Chatpatanasiri, Korsrilabutr,
  Tangchanachaianan, and Kijsirikul}{Chatpatanasiri
  et~al.}{2010}]{chatpatanasiri2010new}
Chatpatanasiri, R., T.~Korsrilabutr, P.~Tangchanachaianan, and B.~Kijsirikul
  (2010).
\newblock A new kernelization framework for {M}ahalanobis distance learning
  algorithms.
\newblock {\em Neurocomputing\/}~{\em 73\/}(10-12), 1570--1579.

\bibitem[\protect\citeauthoryear{Cover and Hart}{Cover and
  Hart}{1967}]{cover1967nearest}
Cover, T. and P.~Hart (1967).
\newblock Nearest neighbor pattern classification.
\newblock {\em IEEE Transactions on Information Theory\/}~{\em 13\/}(1),
  21--27.

\bibitem[\protect\citeauthoryear{Cox, Johnson, and Kafadar}{Cox
  et~al.}{1982}]{cox1982exposition}
Cox, L.~H., M.~M. Johnson, and K.~Kafadar (1982).
\newblock Exposition of statistical graphics technology.
\newblock pp.\  55--56.

\bibitem[\protect\citeauthoryear{Cristianini and Shawe-Taylor}{Cristianini and
  Shawe-Taylor}{2003}]{christianini2003support}
Cristianini, N. and J.~Shawe-Taylor (2003).
\newblock {\em An Introduction to Support Vector Machines and Other
  Kernel-based Learning Methods}.
\newblock Cambridge University Press, London.

\bibitem[\protect\citeauthoryear{Duda, Hart, and Stork}{Duda
  et~al.}{2007}]{duda2006pattern}
Duda, R.~O., P.~E. Hart, and D.~G. Stork (2007).
\newblock {\em Pattern Classification}.
\newblock Wiley, New York.

\bibitem[\protect\citeauthoryear{Dutta and Ghosh}{Dutta and
  Ghosh}{2016}]{dutta2016some}
Dutta, S. and A.~K. Ghosh (2016).
\newblock On some transformations of high dimension, low sample size data for
  nearest neighbor classification.
\newblock {\em Machine Learning\/}~{\em 102\/}(1), 57--83.

\bibitem[\protect\citeauthoryear{Fang, Kotz, and Ng}{Fang
  et~al.}{2018}]{fang2018symmetric}
Fang, K.~W., S.~Kotz, and K.~W. Ng (2018).
\newblock {\em Symmetric Multivariate and Related Distributions}.
\newblock CRC Press, Boca Raton.

\bibitem[\protect\citeauthoryear{Fisher}{Fisher}{1936}]{fisher1936use}
Fisher, R.~A. (1936).
\newblock The use of multiple measurements in taxonomic problems.
\newblock {\em Annals of Eugenics\/}~{\em 7\/}(2), 179--188.

\bibitem[\protect\citeauthoryear{Friedman, Hastie, and Tibshirani}{Friedman
  et~al.}{2010}]{friedman2010regularization}
Friedman, J.~H., T.~Hastie, and R.~Tibshirani (2010).
\newblock Regularization paths for generalized linear models via coordinate
  descent.
\newblock {\em Journal of Statistical Software\/}~{\em 33\/}(1), 1--22.

\bibitem[\protect\citeauthoryear{Ghosh, Chaudhuri, and Murthy}{Ghosh
  et~al.}{2005}]{ghosh2005visualization}
Ghosh, A.~K., P.~Chaudhuri, and C.~A. Murthy (2005).
\newblock On visualization and aggregation of nearest neighbor classifiers.
\newblock {\em IEEE Transactions on Pattern Analysis and Machine
  Intelligence\/}~{\em 27\/}(10), 1592--1602.

\bibitem[\protect\citeauthoryear{Ghosh, Chaudhuri, and Sengupta}{Ghosh
  et~al.}{2006}]{ghosh2006classification}
Ghosh, A.~K., P.~Chaudhuri, and D.~Sengupta (2006).
\newblock Classification using kernel density estimates: multiscale analysis
  and visualization.
\newblock {\em Technometrics\/}~{\em 48\/}(1), 120--132.

\bibitem[\protect\citeauthoryear{Ghosh and Godtliebsen}{Ghosh and
  Godtliebsen}{2012}]{ghosh2012hybrid}
Ghosh, A.~K. and F.~Godtliebsen (2012).
\newblock On hybrid classification using model-assisted posterior estimates.
\newblock {\em Pattern Recognition\/}~{\em 45\/}(6), 2288--2298.

\bibitem[\protect\citeauthoryear{Ghosh and Hall}{Ghosh and
  Hall}{2008}]{ghosh2008error}
Ghosh, A.~K. and P.~Hall (2008).
\newblock On error-rate estimation in nonparametric classification.
\newblock {\em Statistica Sinica\/}~{\em 18\/}(3), 1081--1100.

\bibitem[\protect\citeauthoryear{Hall, Marron, and Neeman}{Hall
  et~al.}{2005}]{hall2005geometric}
Hall, P., J.~S. Marron, and A.~Neeman (2005).
\newblock Geometric representation of high dimension, low sample size data.
\newblock {\em Journal of the Royal Statistical Society: Series B (Statistical
  Methodology)\/}~{\em 67\/}(3), 427--444.

\bibitem[\protect\citeauthoryear{Hand}{Hand}{1982}]{hand1982kernel}
Hand, D.~J. (1982).
\newblock {\em Kernel Discriminant Analysis}.
\newblock Wiley, New York.

\bibitem[\protect\citeauthoryear{Hastie and Tibshirani}{Hastie and
  Tibshirani}{1990}]{hastie1990generalized}
Hastie, T. and R.~Tibshirani (1990).
\newblock {\em Generalized Additive Models. Monographs on Statistics and
  Applied Probability 43}.
\newblock Chapman and Hall, Boca Raton.

\bibitem[\protect\citeauthoryear{Hastie, Tibshirani, and Friedman}{Hastie
  et~al.}{2009}]{hastie2009elements}
Hastie, T., R.~Tibshirani, and J.~H. Friedman (2009).
\newblock {\em The Elements of Statistical Learning: Data Mining, Inference,
  and Prediction}.
\newblock Springer, New York.

\bibitem[\protect\citeauthoryear{Hastie, Tibshirani, and Wainwright}{Hastie
  et~al.}{2015}]{hastie2015statistical}
Hastie, T., R.~Tibshirani, and M.~Wainwright (2015).
\newblock {\em Statistical Learning with Sparsity: the LASSO and
  Generalizations}.
\newblock CRC Press, Boca Raton.

\bibitem[\protect\citeauthoryear{Holmes and Adams}{Holmes and
  Adams}{2002}]{holmes2002probabilistic}
Holmes, C. and N.~Adams (2002).
\newblock A probabilistic nearest neighbour method for statistical pattern
  recognition.
\newblock {\em Journal of the Royal Statistical Society: Series B (Statistical
  Methodology)\/}~{\em 64\/}(2), 295--306.

\bibitem[\protect\citeauthoryear{Holmes and Adams}{Holmes and
  Adams}{2003}]{holmes2003likelihood}
Holmes, C. and N.~Adams (2003).
\newblock Likelihood inference in nearest-neighbour classification models.
\newblock {\em Biometrika\/}~{\em 90\/}(1), 99--112.

\bibitem[\protect\citeauthoryear{Jung and Marron}{Jung and
  Marron}{2009}]{jung2009pca}
Jung, S. and J.~S. Marron (2009).
\newblock {PCA} consistency in high dimension, low sample size context.
\newblock {\em Annals of Statistics\/}~{\em 37\/}(6B), 4104--4130.

\bibitem[\protect\citeauthoryear{Li, Fang, and Zhu}{Li
  et~al.}{1997}]{li1997some}
Li, R.-Z., K.-T. Fang, and L.-X. Zhu (1997).
\newblock Some {QQ} probability plots to test spherical and elliptical
  symmetry.
\newblock {\em Journal of Computational and Graphical Statistics\/}~{\em
  6\/}(4), 435--450.

\bibitem[\protect\citeauthoryear{Mahalanobis}{Mahalanobis}{1936}]{pcm}
Mahalanobis, P. (1936).
\newblock On the generalized distance in statistics.
\newblock {\em Proceedings of the National Institute of Sciences, India\/}~{\em
  2}, 49--55.

\bibitem[\protect\citeauthoryear{Mika, Ratsch, Weston, Scholkopf, and
  Mullers}{Mika et~al.}{1999}]{mika1999fisher}
Mika, S., G.~Ratsch, J.~Weston, B.~Scholkopf, and K.-R. Mullers (1999).
\newblock Fisher discriminant analysis with kernels.
\newblock In {\em Neural Networks for Signal Processing IX: Proceedings of the
  1999 IEEE Signal Processing Society Workshop}, pp.\  41--48.

\bibitem[\protect\citeauthoryear{Murray, Liao, Stankovic, Stankovic,
  Hauxwell-Baldwin, Wilson, Coleman, Kane, and Firth}{Murray
  et~al.}{2015}]{murray2015data}
Murray, D., J.~Liao, L.~Stankovic, V.~Stankovic, R.~Hauxwell-Baldwin,
  C.~Wilson, M.~Coleman, T.~Kane, and S.~Firth (2015).
\newblock A data management platform for personalised real-time energy
  feedback.
\newblock In {\em 8th International Conference on Energy Efficiency in Domestic
  Appliances and Lighting}.

\bibitem[\protect\citeauthoryear{Pal, Mondal, and Ghosh}{Pal
  et~al.}{2016}]{pal2016high}
Pal, A.~K., P.~K. Mondal, and A.~K. Ghosh (2016).
\newblock High dimensional nearest neighbor classification based on mean
  absolute differences of inter-point distances.
\newblock {\em Pattern Recognition Letters\/}~{\em 74}, 1--8.

\bibitem[\protect\citeauthoryear{Ripley}{Ripley}{2008}]{ripley2007pattern}
Ripley, B.~D. (2008).
\newblock {\em Pattern Recognition and Neural Networks}.
\newblock Cambridge University Press, London.

\bibitem[\protect\citeauthoryear{Rousseeuw and Driessen}{Rousseeuw and
  Driessen}{1999}]{rousseeuw1999fast}
Rousseeuw, P.~J. and K.~V. Driessen (1999).
\newblock A fast algorithm for the minimum covariance determinant estimator.
\newblock {\em Technometrics\/}~{\em 41\/}(3), 212--223.

\bibitem[\protect\citeauthoryear{Roverso}{Roverso}{2000}]{roverso2000multivariate}
Roverso, D. (2000).
\newblock Multivariate temporal classification by windowed wavelet
  decomposition and recurrent neural networks.
\newblock In {\em 3rd ANS International Topical Meeting on Nuclear Plant
  Instrumentation, Control and Human-Machine Interface}, Volume~20.

\bibitem[\protect\citeauthoryear{Ruiz and L{\'o}pez-de Teruel}{Ruiz and
  L{\'o}pez-de Teruel}{2001}]{ruiz2001nonlinear}
Ruiz, A. and P.~E. L{\'o}pez-de Teruel (2001).
\newblock Nonlinear kernel-based statistical pattern analysis.
\newblock {\em IEEE Transactions on Neural Networks\/}~{\em 12\/}(1), 16--32.

\bibitem[\protect\citeauthoryear{Sarkar and Ghosh}{Sarkar and
  Ghosh}{2019}]{sarkar2019perfect}
Sarkar, S. and A.~K. Ghosh (2019).
\newblock On perfect clustering of high dimension, low sample size data.
\newblock {\em IEEE Transactions on Pattern Analysis and Machine
  Intelligence\/}~{\em 42\/}(9), 2257--2272.

\bibitem[\protect\citeauthoryear{Sch{\"o}lkopf and Smola}{Sch{\"o}lkopf and
  Smola}{2002}]{smola1998learning}
Sch{\"o}lkopf, B. and A.~J. Smola (2002).
\newblock {\em Learning with Kernels: Support Vector Machines, Regularization,
  Optimization, and Beyond}.
\newblock The MIT Press, Cambridge.

\bibitem[\protect\citeauthoryear{Smith, Everhart, Dickson, Knowler, and
  Johannes}{Smith et~al.}{1988}]{smith1988using}
Smith, J.~W., J.~E. Everhart, W.~Dickson, W.~C. Knowler, and R.~S. Johannes
  (1988).
\newblock Using the {ADAP} learning algorithm to forecast the onset of diabetes
  mellitus.
\newblock In {\em Proceedings of the Annual Symposium on Computer Application
  in Medical Care}, pp.\  261--265.

\bibitem[\protect\citeauthoryear{Van~Aelst and Rousseeuw}{Van~Aelst and
  Rousseeuw}{2009}]{van2009minimum}
Van~Aelst, S. and P.~Rousseeuw (2009).
\newblock Minimum volume ellipsoid.
\newblock {\em Wiley Interdisciplinary Reviews: Computational
  Statistics\/}~{\em 1\/}(1), 71--82.

\bibitem[\protect\citeauthoryear{Wood}{Wood}{2017}]{wood2017generalized}
Wood, S.~N. (2017).
\newblock {\em Generalized Additive Models: An Introduction with R}.
\newblock CRC Press, Boca Raton.

\bibitem[\protect\citeauthoryear{Yata and Aoshima}{Yata and
  Aoshima}{2020}]{yata2020geometric}
Yata, K. and M.~Aoshima (2020).
\newblock Geometric consistency of principal component scores for
  high-dimensional mixture models and its application.
\newblock {\em Scandinavian Journal of Statistics\/}~{\em 47\/}(3), 899--921.

\bibitem[\protect\citeauthoryear{Yeh, Yang, and Ting}{Yeh
  et~al.}{2009}]{yeh2009knowledge}
Yeh, I.-C., K.-J. Yang, and T.-M. Ting (2009).
\newblock Knowledge discovery on {RFM} model using {B}ernoulli sequence.
\newblock {\em Expert Systems with Applications\/}~{\em 36\/}(3), 5866--5871.

\end{thebibliography}

\end{document}